%% file: main.tex
\documentclass{article} 
\PassOptionsToPackage{numbers,compress}{natbib}
\usepackage{iclr2027_conference,times}
\iclrfinalcopy
\setcitestyle{numbers,square,comma}
\usepackage{titlesec}
\titlespacing*{\section}{0pt}{1.5ex plus 0.5ex minus 0.2ex}{0.8ex}
\titlespacing*{\subsection}{0pt}{1.2ex plus 0.5ex minus 0.2ex}{0.5ex}
\titlespacing*{\subsubsection}{0pt}{1ex plus 0.5ex minus 0.2ex}{0.3ex}
\usepackage[utf8]{inputenc}
\usepackage[T1]{fontenc}
\usepackage{microtype}

\input{math_commands.tex}

\usepackage{algorithm}
\usepackage{algorithmic}
\usepackage{makecell}
\usepackage{hyperref}
\usepackage{url}
\usepackage{graphicx}
\usepackage{booktabs}
\usepackage{multirow}
\usepackage{graphicx}
\usepackage{subcaption}
\usepackage{amsthm}
\newtheorem{thm}{Theorem}[section]
\newtheorem{prp}[thm]{Proposition}
\newtheorem{lmm}[thm]{Lemma}
\newtheorem{asm}[thm]{Assumption}
\title {Bayesian Offline Multi-Agent Reinforcement Learning with Latent World Models}

\author{
Xi Xiong$^{1,3}$ \thanks{Work done while visiting the University of Oxford.} \quad Lu Liu$^{3}$ \quad Yiming Li$^{3}$ \quad Jakob Foerster$^{1,2}$ \\[0.5em]
\normalfont $^{1}$\,British Open-ended Learning and Discovery Lab (BOLD) \\[0.2em]
\normalfont $^{2}$\,University of Oxford \\[0.2em]
\normalfont $^{3}$\,Tongji University
}

\begin{document}

\maketitle
\fancyhead[L]{Preprint, October 2026.}

\begin{abstract}
Model-based offline multi-agent reinforcement learning (MARL) provides a promising approach to multi-agent decision-making under partial observability, where the global state is unavailable from offline data. However, limited dataset coverage can lead to substantial out-of-distribution (OOD) model errors, which may mislead policy optimization and degrade downstream performance. Existing model-based offline MARL methods typically rely on access to global states or mitigate model errors through pessimistic or conservative policy optimization, leaving the reliability and refinement of learned world models less explored. We propose a Bayesian offline MARL framework that infers latent global states from decentralized observation-action histories and iteratively refines a posterior world model using only offline data. Building on posterior information loss, we extend this criterion to the multi-agent setting as multi-agent posterior information loss (MPIL), which provides a reliability measure for world model refinement and offline evaluation. Theoretically, we establish an upper bound on policy regret that is monotonic in MPIL and further characterize its dependence on the size and state–joint-action coverage of the offline dataset. Experiments on partially observable multi-agent variants of D4RL MuJoCo tasks demonstrate competitive policy performance across diverse offline data conditions, while MPIL-guided refinement improves world model prediction and downstream policy performance. These results support a principled connection between world model reliability, offline data coverage, and policy performance in partially observable offline MARL.
\end{abstract}

\input{./sections/01_Introduction}
\input{./sections/02_Preliminaries}

\input{./sections/03_Reviews}
\input{./sections/04_Bayesian}
\input{./sections/05_Regret}
\input{./sections/07_Experiments}

\input{./sections/08_Conclusion}

\subsection*{AI use statement}

In this work, we used generative AI tools to assist with language editing, manuscript organization and presentation, clarification and checking of mathematical derivations and theoretical exposition, and interpretation of experimental results. Generative AI tools were also used to assist with LaTeX formatting and technical writing. All AI-assisted theoretical statements, mathematical derivations, and experimental analyses were independently reviewed and verified by the authors. Experimental data, results, and conclusions were produced and validated by the authors, and all AI-assisted text was manually reviewed and revised. Generative AI tools were not used to generate or alter experimental data or results. We take full responsibility for the final content of this work, including all text, claims, and artifacts produced with the aid of generative AI.




\subsection*{Ethics statement}

This study does not involve human participants, personally identifiable information, or sensitive private data. 



\subsection*{Reproducibility statement}

We provide detailed descriptions of the proposed framework, model formulation, training procedure, and experimental settings in the main paper and appendix. The key hyperparameters, network configurations, offline data generation process, baseline implementations, and evaluation metrics are reported to facilitate reproduction of the experimental results. Complete derivations and proofs of the theoretical results are also provided in the appendix.




\bibliography{iclr2027_conference}
\bibliographystyle{unsrtnat}

\appendix
\input{./sections/09_Appendix}

\end{document}

%% file: math_commands.tex
\usepackage{amsmath,amsfonts,bm}

\def\eqref#1{equation~\ref{#1}}

\def\1{\bm{1}}

\DeclareMathAlphabet{\mathsfit}{\encodingdefault}{\sfdefault}{m}{sl}
\SetMathAlphabet{\mathsfit}{bold}{\encodingdefault}{\sfdefault}{bx}{n}



%% file: sections/01_Introduction.tex
\section{Introduction}
Multi-agent reinforcement learning (MARL) has achieved remarkable success in complex cooperative tasks \citep{lowe2017multi}. However, conventional MARL typically requires extensive online interaction during training, which can incur substantial safety risks and trial-and-error costs when a sufficiently accurate simulator is unavailable. Offline MARL addresses this limitation by learning policies exclusively from historical datasets and has been widely studied in fields such as healthcare, industrial automation, and autonomous driving \citep{emerson2023offline, ma2024improving, hu2025coordinating, lee2025scenario}. Nevertheless, in many practical partially observable multi-agent systems, the global state is inaccessible, leaving agents with only decentralized local observations. This information deficiency undermines centralized policy learning, since local observations may be insufficient to characterize the underlying environment and support reliable coordination \citep{xu2022side}. Therefore, how to learn effective policies solely from offline data when global states are unavailable remains a fundamental challenge in offline MARL.

Model-based offline RL offers a natural paradigm for fully offline policy learning and evaluation, and its core ideas have also been extended to multi-agent settings. By learning environment dynamics from offline data, a world model can serve as a surrogate environment for policy optimization and evaluation. However, limited data coverage can make the learned model unreliable, causing policy optimization to exploit model errors in out-of-distribution (OOD) regions. Existing methods mainly address this issue through pessimistic or conservative policy optimization \citep{yu2020mopo,kidambi2020morel,yu2021combo,barde2024moma}. While effective, these methods mainly mitigate the impact of model errors during policy optimization, rather than improving the world model itself; moreover, optimizing policies against an unreliable model can incur substantial and avoidable computational cost. Moreover, offline RL lacks a principled criterion linking world model reliability to policy performance and safety. SOReL \citep{fellows2026fully} takes an important step by introducing posterior information loss (PIL), which links world model discrepancy to policy regret, although it is limited to the single-agent setting and does not employ PIL for world model refinement.

To address these challenges, we propose a latent world model framework for fully offline MARL under decentralized partial observability. The framework infers latent global states from agents' local observation-action histories and models posterior uncertainty over transition, reward, and observation dynamics. Building on this model, we introduce multi-agent posterior information loss (MPIL) to quantify world model reliability and further integrate empirical MPIL into world model learning as a fine-tuning objective and reliability criterion. Starting from an initially trained world model, the framework iteratively refines the model until convergence, after which the resulting world model is used for policy learning and fully offline regret estimation, forming a closed loop of world model refinement, policy learning, and performance assessment without further interaction with the true environment. 

Within this framework, we prove that the true regret is upper-bounded by a monotonic function of MPIL, and further characterize its dependence on offline dataset size and state–joint-action coverage. These results establish a unified theoretical connection among world model reliability, offline data availability, and policy-performance guarantees. We validate the proposed theory and framework on partially observable multi-agent MuJoCo \citep{dewitt2020deep} benchmarks, where MPIL-guided world model refinement improves model prediction accuracy and policy performance, and the regret estimated from the refined world model follows a trend consistent with the true regret as the amount of offline data increases.

%% file: sections/02_Preliminaries.tex
\section{Preliminaries}
\label{sec:preliminaries}

The decentralized partially observable Markov decision process (Dec-POMDP) provides a general framework for modeling cooperative multi-agent decision-making problems in which each agent has access only to local observations, while all agents share a common system-level objective \citep{bernstein2002complexity,bellman1956problem}. A Dec-POMDP is formally defined by the tuple $
    G=\left\langle \mathcal{I}, S,\left\{A_{i}\right\}, T, R,\left\{Z_{i}\right\},\left\{O_{i}\right\}, \gamma\right\rangle$.
Here, $\mathcal{I}=\{1,2,\ldots,I\}$ denotes the finite set of agents, and $S$ denotes the environment state space with an initial state distribution $P_{0}$. 
At each time step $t=0,1,2,\ldots$, the environment occupies a global state $s_t\in S$, and each agent $i \in \mathcal{I}$ selects an action $a_{i,t}$ from its action space $A_i$ according to its policy $\pi_i$.
The individual actions form the joint action $\tilde{a}_t=(a_{1,t},\ldots,a_{I,t})$.
The transition function $T(s_{t+1}|s_t,\tilde{a}_t)=\operatorname{Pr}\left(s_{t+1}\mid s_t,\tilde{a}_t\right)$ characterizes the stochastic evolution of the environment. The reward function is defined as $R(r_t|s_t,\tilde{a}_t)=\operatorname{Pr}(r_t\mid s_t,\tilde{a}_t)$, where the resulting global reward is shared among all agents.
Under partial observability, each agent $i$ receives a local observation $z_{i,t}$ from its observation space $Z_i$ rather than directly accessing the global state $s_t$, where the observation function $O_i(z_{i,t+1}|s_t,\tilde{a}_t)=\operatorname{Pr}\left(z_{i,t+1}\mid s_t,\tilde{a}_t\right)$ characterizes the generation of local observations.
The policy $\pi_i$ maps the local history $h_{i,t}=(z_{i,0},a_{i,0},\ldots,a_{i,t-1},z_{i,t})$ to an action distribution $\pi_i(a_{i,t}\mid h_{i,t})$, and the joint policy is denoted by $\tilde{\pi}=(\pi_1,\ldots,\pi_I)$.
The discount factor $\gamma\in[0,1)$ determines the relative importance of future rewards. 

To model partially observable multi-agent environments, we construct a dynamics model $D_{T,R}$ that jointly captures state transitions and shared rewards, together with an observation model $\tilde{O}=\{O_i\}_{i\in \mathcal{I}}$ for local-observation generation. The resulting world model is defined as $\mathcal{M}_{\mathrm{WM}}=(D_{T,R},\tilde{O})$. For model learning, $T$ and $R$ are jointly parameterized through a transition-reward distribution $P_{T,R}$, whereas the observation distribution are represented separately by $P_{\tilde O}$. In this work, the true global state $s_t$ is unavailable during both training and execution, and the offline dataset therefore contains only local observations, joint actions, and shared rewards, $\mathcal{D}_N=\{(\tilde{z}_t,\tilde{a}_t,r_t,\tilde{z}_{t+1})\}_{t=0}^{N-1}$, where $N$ is the number of transitions in the offline dataset.

%% file: sections/03_Reviews.tex
\section{Related Work}
Existing model-free offline RL methods mitigate distribution shift through policy constraints, conservative value estimation, and pessimistic learning \citep{kumar2020conservative,fujimoto2019off,kumar2019stabilizing,fujimoto2021minimalist,bai2022pessimistic}, with multi-agent extensions further addressing coordination through CTDE, decentralized optimization, and global-to-local regularization 
\citep{yang2021believe,pan2022plan,wang2023offline,shao2023counterfactual}. However, these methods do not explicitly model environment dynamics, limiting their ability to explicitly reason about long-horizon consequences through learned environment dynamics. Model-based offline RL provides an alternative research direction by learning environment dynamics and leveraging imagined trajectories for policy optimization. Existing methods mainly mitigate the impact of model errors on policy optimization through pessimistic modeling \citep{yu2020mopo,kidambi2020morel,wu2024ocean}, uncertainty estimation \citep{qiao2025sumo,sun2023model}, and conservative learning \citep{yu2021combo}. Most of these methods assume that the global state required for dynamics learning is directly available, including recent model-based extensions to offline MARL \citep{barde2024moma}. In partially observable multi-agent environments, however, the global state may not be directly accessible from offline data. A common workaround is to concatenate agents' local observations as a surrogate global representation \citep{lobos2022ma,barde2024moma}, while another line of work learns latent state representations from observation-action histories to capture hidden system dynamics \citep{gregor2019shaping,venugopal2023mabl,shi2025gawm}. Although these approaches enable dynamics modeling without direct access to the true global state, how to assess and refine the learned world model solely from fixed offline data remains underexplored. 

Fully offline policy assessment has been studied from several perspectives. Off-policy evaluation (OPE) estimates target-policy performance from logged data using importance weighting, value-function estimation, or learned dynamics models \citep{thomas2016dataefficient,fu2021benchmarks,uehara2026review}, while offline policy-selection methods use such estimates or alternative validation criteria to rank candidate policies and tune hyperparameters without environment interaction \citep{zhang2021towards,da2024probabilistic}. These approaches primarily assess or select policies and value functions, rather than characterizing whether a learned world model is sufficiently reliable for downstream policy learning. More recently, SOReL \citep{fellows2026fully} introduces PIL to connect posterior world model discrepancy with policy regret, enabling fully offline model assessment and tuning. However, current PIL-based approaches are mainly developed for single-agent scenarios and do not explicitly leverage PIL to update world model parameters.

In this work, we develop an MPIL-guided latent world model framework for partially observable offline MARL, which infers latent system states from local observation-action histories without access to global states. MPIL is used to assess and iteratively refine the world model fully offline, while its connection to policy regret enables principled offline evaluation of multi-agent policy performance.

%% file: sections/04_Bayesian.tex
\section{Methodology}
\label{sec:bayesian_world_model}

In this section, we present a Bayesian world model-based offline MARL framework for partially observable environments, as illustrated in Fig.~\ref{fig:mpil_comparison}. To obtain a system-level representation for world model learning and cooperative policy optimization, the proposed framework first infers a latent global state $\hat{s}_t$ from offline local observation-action histories and then learns a Bayesian world model. The world model is further optimized using an empirical multi-agent posterior information loss (MPIL) objective and assessed using an MPIL-based reliability criterion. Once the criterion is satisfied, the world model is used for cooperative policy optimization under the CTDE paradigm \citep{lowe2017multi,rashid2020monotonic}.

\begin{figure}[htbp]
    \centering
    \begin{subfigure}[t]{0.4\textwidth}
        \centering
        \includegraphics[width=\linewidth]{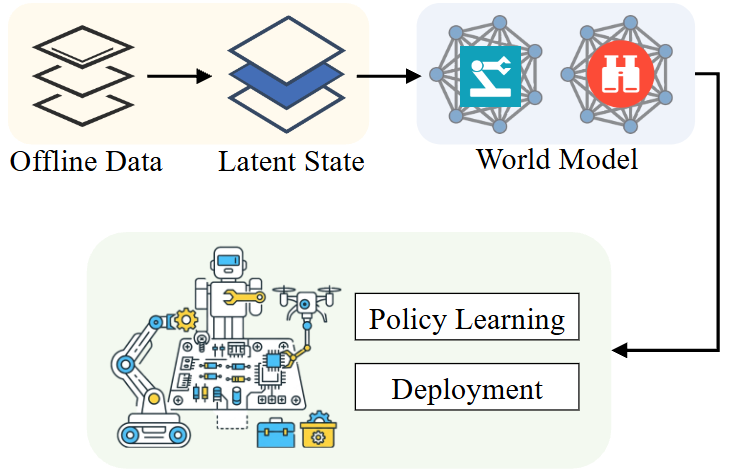}
        \caption{Traditional Model-Based Offline MARL}
        \label{fig:conventional_mborl}
    \end{subfigure}
    \hfill
    \begin{subfigure}[t]{0.58\textwidth}
        \centering
        \includegraphics[width=\linewidth]{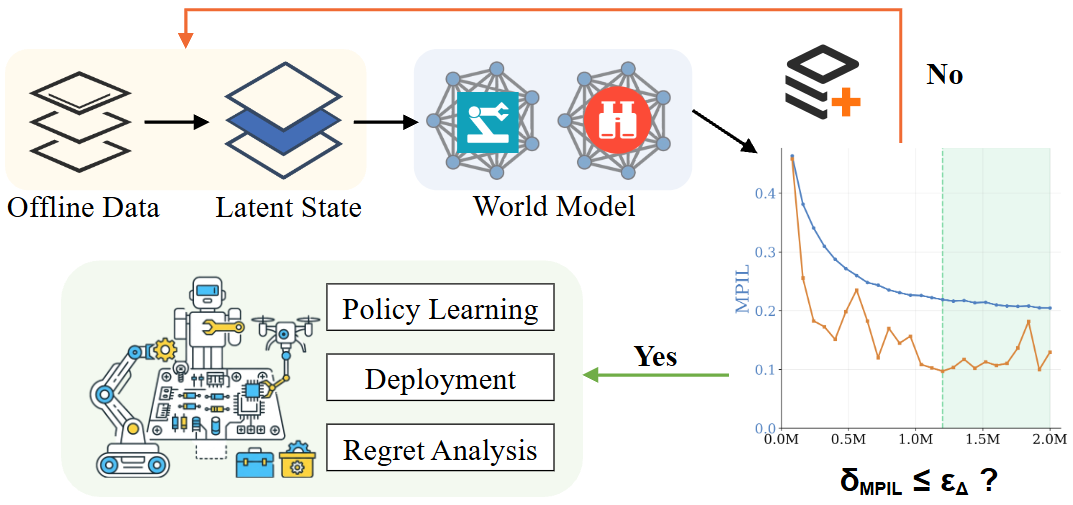}
        \caption{MPIL-Guided Model-Based Offline MARL}
        \label{fig:mpil_mborl}
    \end{subfigure}
    \caption{Comparison of Traditional and MPIL-Guided Model-Based Offline MARL Frameworks.}
    \label{fig:mpil_comparison}
\end{figure}

\subsection{Latent Global State Inference}
\label{subsec:state_inference}

Although concatenating local observations provides a straightforward approximation of global information, the resulting representation can become high-dimensional as the number of agents increases and may still omit information that is unobserved by all agents \citep{yu2022surprising}.
We therefore infer a compact history-dependent latent global state from joint observation-action histories \citep{venugopal2023mabl}.
Let $\tilde{z}_t$ and $\tilde{a}_t$ denote the joint observation and joint action of all agents at time step $t$, respectively. A gated recurrent unit (GRU) is employed as a recurrent history encoder to infer the latent state $\hat{s}_t$ \citep{igl2018deep},
\begin{align}
\hat{s}_t = \mathrm{GRU}_{\xi} \left( \hat{s}_{t-1}, [\tilde{z}_t,\tilde{a}_{t-1}] \right),
\end{align}
where $\xi$ denotes the trainable parameters of the GRU encoder. 

To regularize the inferred latent representation and encourage it to preserve observation-relevant information, we introduce an auxiliary observation decoder that reconstructs the joint observation from the latent state $\check{z}_t = D_{\psi}\left(\hat{s}_t\right)$, where $\psi$ denotes the parameters of the auxiliary decoder.
The encoder and auxiliary decoder are jointly trained by minimizing the observation reconstruction loss,
\begin{align}
\mathcal{L}_{\mathrm{aux}}(\xi,\psi) = \frac{1}{N} \sum_{t=0}^{N-1} \left\| \tilde{z}_t-\check{z}_t \right\|_2^2.
\label{eq:reconstruction_loss}
\end{align}

After representation learning, the offline data are mapped into the latent state space. The resulting state-augmented dataset is $\hat{\mathcal{D}}_N= \left\{\left(\hat{s}_t,\tilde{z}_t,\tilde{a}_t, r_t,\hat{s}_{t+1},\tilde{z}_{t+1}\right)\right\}_{t=0}^{N-1}$.

\subsection{World Model Formulation}
\label{subsec:bayesian_wm_formulation}

In this subsection, we adopt a Bayesian formulation for the transition-reward parameters $\theta$ and observation parameters $\tilde{\omega}$ \citep{folch2024transition,duff2002optimal}.

Based on the inferred latent global state, the parameterized transition-reward model is defined as $P_{T,R}=p(\hat{s}_{t+1},r_t\mid\hat{s}_t,\tilde{a}_t,\theta)$, where $\theta\in\Theta\subseteq\mathbb{R}^{d}$. For each agent $i$, the observation model is defined as $P_{O_i}=p(z_{i,t+1}\mid\hat{s}_t,\tilde{a}_t,\omega_i)$, where $\omega_i\in\Omega_i\subseteq\mathbb{R}^{d}$. We assume conditional independence of local observations given $(\hat{s}_t,\tilde{a}_t)$, yielding $P_{\tilde{O}}=p(\tilde{z}_{t+1}\mid\hat{s}_t,\tilde{a}_t,\tilde{\omega})=\prod_{i=1}^{I}p(z_{i,t+1}\mid\hat{s}_t,\tilde{a}_t,\omega_i)$. We collect $\tilde{\omega}=(\omega_1,\ldots,\omega_I)\in\tilde{\Omega}$, where $\tilde{\Omega}=\Omega_1\times\cdots\times\Omega_I$, and use the factorized prior $p(\tilde{\omega})=\prod_{i=1}^{I}p(\omega_i)$. Prior distributions $P_{\Theta}$ and $P_{\Omega_i}$ encode the initial beliefs over the transition-reward and observation models, respectively. Given $\hat{\mathcal{D}}_N$, Bayes' rule \citep{ross2011bayesian} yields the joint posterior over the complete world model parameters,
\begin{align}
p(\theta,\tilde{\omega}\mid\hat{\mathcal{D}}_N) \propto p(\theta) \prod_{i=1}^{I} p(\omega_i) \prod_{t=0}^{N-1} \left[ p\left( \hat{s}_{t+1},r_t \mid \hat{s}_t,\tilde{a}_t,\theta \right) \prod_{i=1}^{I} p\left( z_{i,t+1} \mid \hat{s}_t,\tilde{a}_t,\omega_i \right) \right].
\end{align}

Together, these components define $P_{\Theta,\tilde{\Omega}}(\cdot\mid\hat{\mathcal{D}}_N)$, which characterizes posterior uncertainty over the complete world model. We train the world model using the corresponding posterior-averaged negative log-likelihood (NLL), defined as
\begin{align}
\label{ideal_NLL}
\mathcal{L}_{\mathrm{NLL}}
=
\mathbb{E}_{(\hat{s},\tilde{a})\sim\nu_{\hat{\mathcal{D}}_N}}
\mathbb{E}_{(\theta,\tilde{\omega})\sim
P_{\Theta,\tilde{\Omega}}(\cdot\mid\hat{\mathcal{D}}_N)}
\Big[
&
\mathbb{E}_{(s',r)\sim P_{T,R}^{\star}(\cdot\mid\hat{s},\tilde{a})}
\left[
-\log P_{T,R}(s',r\mid\hat{s},\tilde{a},\theta)
\right]
\nonumber\\
&+
\mathbb{E}_{\tilde{o}'\sim
P_{\tilde{O}}^{\star}(\cdot\mid\hat{s},\tilde{a})}
\left[
-\log P_{\tilde{O}}(\tilde{o}'\mid\hat{s},\tilde{a},\tilde{\omega})
\right]
\Big].
\end{align}
This objective evaluates how well posterior world model hypotheses assign likelihood to transitions generated by the true environment, and serves as the world model learning objective.

To evaluate world model reliability before policy learning, we extend posterior information loss (PIL) \citep{fellows2026fully} to the multi-agent setting and further reformulate it as a dataset-weighted criterion. Specifically, we define multi-agent posterior information loss (MPIL) as the posterior-averaged discrepancy between the true and learned world model distributions over state-action pairs covered by the offline dataset:
\begin{align}
\label{PIL}
\mathcal{I}_{\hat{\mathcal{D}}_N}
=
\mathbb{E}_{(\theta,\tilde{\omega})\sim
P_{\Theta,\tilde{\Omega}}(\cdot\mid\hat{\mathcal{D}}_N)}
\mathbb{E}_{(\hat{s}_j,\tilde{a}_j)\sim
\nu_{\hat{\mathcal{D}}_N}}
\Big[
&D_{\mathrm{KL}}\left(
P_{T,R}^{\star}(\cdot\mid\hat{s}_j,\tilde{a}_j)
\|
P_{T,R}(\cdot\mid\hat{s}_j,\tilde{a}_j,\theta)
\right)
\nonumber\\
&+
D_{\mathrm{KL}}\left(
P_{\tilde{O}}^{\star}(\cdot\mid\hat{s}_j,\tilde{a}_j)
\|
P_{\tilde{O}}(\cdot\mid\hat{s}_j,\tilde{a}_j,\tilde{\omega})
\right)
\Big].
\end{align}
Here, $P_{T,R}^{\star}$ and $P_{\tilde{O}}^{\star}$ denote the transition--reward and joint-observation distributions of the true environment, respectively.
$\nu_{\hat{\mathcal{D}}_N}$ denotes the underlying state-action distribution associated with the offline dataset $\hat{\mathcal{D}}_N$; in practice, expectations with respect to $\nu_{\hat{\mathcal{D}}_N}$ are approximated by uniformly sampling transitions from the dataset.
The following proposition establishes the theoretical connection between MPIL and the NLL objective.

\begin{prp}
\label{proposition_nll_mpil}
For a fixed offline data distribution $\nu_{\hat{\mathcal{D}}_N}$,
\begin{align}
\label{MPIL_NLL}
\mathcal{I}_{\hat{\mathcal{D}}_N}
=
\mathcal{L}_{\mathrm{NLL}}
-
\mathcal{H}^{\star}_{\hat{\mathcal{D}}_N},
\end{align}
where $\mathcal{H}^{\star}_{\hat{\mathcal{D}}_N} = \mathbb{E}_{(\hat{s},\tilde{a})\sim\nu_{\hat{\mathcal{D}}_N}} \Big[
H\!\left( P_{T,R}^{\star}(\cdot\mid\hat{s},\tilde{a}) \right)
+ H\!\left( P_{\tilde{O}}^{\star}(\cdot\mid\hat{s},\tilde{a}) \right) \Big]$  denotes the conditional entropy of the true environment under the offline data distribution.
\end{prp}

The proof, based on the KL--cross-entropy identity, is provided in Appendix~\ref{appendix:proof_prp1}. This relation provides a tractable way to estimate MPIL through the posterior-averaged NLL. Since $\mathcal{H}^{\star}_{\hat{\mathcal{D}}_N}$ is independent of the world model parameters, minimizing the posterior-averaged NLL is equivalent to minimizing MPIL, and NLL changes across refinement stages directly track changes in MPIL. 

Exact posterior inference is generally intractable. We therefore employ an ensemble of $K$ world models,
$\{\mathcal{M}_k\}_{k=1}^{K}$, where
$\mathcal{M}_k=(T_{\theta_k},\tilde{O}_{\tilde{\omega}_k})$,
to approximate the joint parameter posterior as
\begin{align}
P_{\Theta,\tilde{\Omega}}
\left(\cdot \mid \hat{\mathcal{D}}_N\right)
\approx
\widehat{P}_{\Phi}(\cdot)
:=
\frac{1}{K}
\sum_{k=1}^{K}
\delta_{(\theta_k,\tilde{\omega}_k)}(\cdot),
\label{eq:ensemble_posterior}
\end{align}
where $\Phi=\{(\theta_k,\tilde{\omega}_k)\}_{k=1}^{K}$ denotes the collection of ensemble parameters. 
Each parameter pair represents a candidate world model hypothesis, and the ensemble members are independently initialized to encourage model diversity.

Under this ensemble approximation, the posterior expectation in Eq.~(\ref{ideal_NLL}) is replaced by an empirical average over the $K$ world models. 
For an offline subset $\hat{\mathcal{D}}$, the empirical NLL is computed as
\begin{align}
\label{empirical_NLL}
\widehat{\mathcal{L}}_{\mathrm{NLL}}
(\Phi;\hat{\mathcal{D}})
=
-\frac{1}{K|\hat{\mathcal{D}}|}
\sum_{k=1}^{K}
\sum_{\tau \in\hat{\mathcal{D}}}
\Big[
\log P_{T,R}
(s'_i,r_i\mid\hat{s}_i,\tilde{a}_i,\theta_k)
+
\log P_{\tilde{O}}
(\tilde{o}'_i\mid\hat{s}_i,\tilde{a}_i,\tilde{\omega}_k)
\Big],
\end{align}
where $\tau=(\hat{s}_t,\tilde{a}_t,\hat{s}_{t+1},r_t,\tilde{o}_{t+1})$ denotes an offline transition in $\hat{\mathcal{D}}$.
We use $\widehat{\mathcal{L}}_{\mathrm{NLL}}$ as the practical world model training loss and as an empirical surrogate for monitoring MPIL during iterative world model refinement.

\subsection{Learning an Approximately Bayes-Optimal Policy}
\label{subsec:4.3}
We denote the parameters of the initially learned world model ensemble by $\Phi^{(1)}$. We assess world model convergence using the change in MPIL between consecutive refinement stages. A fixed held-out validation set $\hat{\mathcal{D}}_{\mathrm{val}}$ is used throughout refinement, and we define
\begin{align}
\label{delta_MPIL}
\delta_{\mathrm{MPIL}}^{(m+1)}
:=
\left|
\mathcal{I}_{\hat{\mathcal{D}}_{\mathrm{val}}}
\left(\Phi^{(m+1)}\right)
-
\mathcal{I}_{\hat{\mathcal{D}}_{\mathrm{val}}}
\left(\Phi^{(m)}\right)
\right|.
\end{align}
By Proposition~\ref{proposition_nll_mpil}, the entropy term remains constant under the same validation distribution, so $\delta_{\mathrm{MPIL}}^{(m+1)}$ equals the corresponding change in posterior-averaged NLL and can be computed entirely offline using the empirical estimator in Eq.~(\ref{empirical_NLL}).

Let $\{\mathcal{D}_m\}_{m=1}^{R}$ denote the raw data subsets incorporated during refinement, and let $\hat{\mathcal{D}}^{(m)}$ denote the cumulative latent-state-augmented training data constructed from the first $m$ subsets. We use $\Phi^{(m)}$ to denote the ensemble parameters at refinement stage $m$, with the initially learned ensemble indexed by $m=1$. 
At stage $m$, the next subset $\mathcal{D}_{m+1}$ is incorporated to construct $\hat{\mathcal{D}}^{(m+1)}$, and the ensemble is warm-started from $\Phi^{(m)}$ and refined by minimizing the empirical NLL objective
\begin{equation}
\Phi \leftarrow \Phi
-\eta\nabla_{\Phi}
\hat{\mathcal{L}}_{\mathrm{NLL}}
\left(\Phi;\hat{\mathcal{D}}_{\mathrm{train}}^{(m+1)}\right),
\qquad
\left.\Phi\right|_{\mathrm{start}}=\Phi^{(m)},
\label{eq}
\end{equation}
where $\eta$ denotes the learning rate. Repeating this update yields $\Phi^{(m+1)}$, which induces the refined empirical posterior $P_{\Theta,\tilde{\Omega}}^{(m+1)}$.
After each refinement stage, both $\Phi^{(m)}$ and $\Phi^{(m+1)}$ are evaluated on the validation set to compute $\delta_{\mathrm{MPIL}}^{(m+1)}$. 
The refinement procedure terminates when $
\delta_{\mathrm{MPIL}}^{(m+1)}
\leq
\epsilon_{\Delta},
$ where $\epsilon_{\Delta}$ is a convergence tolerance, or when all $R$ data subsets have been incorporated. 
The terminating refinement stage is denoted by $m^{\star}$, and its empirical posterior $P_{\Theta,\tilde{\Omega}}^{(m^{\star})}$ is subsequently used as the surrogate environment distribution for policy learning.

To generate each imagined trajectory, we sample
$(\theta,\tilde{\omega})\sim P_{\Theta,\tilde{\Omega}}^{(m^{\star})}$
and roll out the joint policy $\tilde{\pi}$ in the corresponding world model.
Let $P_{\infty,\tilde{\pi}}(\theta,\tilde{\omega})$ denote the induced imagined-trajectory distribution, and let $\check{h}_{\infty}$ denote the corresponding interaction history. The Bayesian policy objective is defined as
\begin{align}
\label{Bayesian_policy_objective}
J_{\mathrm{Bayes}}^{\tilde{\pi}}
\left(
P_{\Theta,\tilde{\Omega}}^{(m^{\star})}
\right)
=
\mathbb{E}_{\substack{
(\theta,\tilde{\omega})\sim P_{\Theta,\tilde{\Omega}}^{(m^{\star})},
\check{h}_{\infty}\sim
P_{\infty,\tilde{\pi}}(\theta,\tilde{\omega})
}}
\left[
\sum_{t=0}^{\infty}
\gamma^t\check{r}_t
\right].
\end{align}
We optimize this objective using SAC \citep{haarnoja2018soft} under CTDE \citep{lowe2017multi}, where the centralized critic conditions on inferred latent states and joint actions, while each decentralized actor uses only its local observation history. The resulting policy is then evaluated fully offline using the regret. The complete procedure is summarized in Algorithm~\ref{alg:pil_marl} in the Appendix.

%% file: sections/05_Regret.tex
\section{Regret Analysis}
This section relates MPIL to policy performance by bounding the Bayes-optimal policy regret as a monotone function of MPIL and characterizing its dependence on offline data size and coverage.
\subsection{Controlling Regret with the MPIL}
For a finite $N$, the true regret measures the difference in expected return between the Bayes-optimal policy $\tilde{\pi}_{\operatorname{Bayes}}^{\star}$ learned from the world model and the optimal policy $\tilde{\pi}^{\star}$ in the true environment $\mathcal{M}^{\star}$,
\begin{equation}
\label{Regret}
    \operatorname{Regret}\left(\mathcal{M}^{\star}, \hat{\mathcal{D}}_{N}\right):=J^{\tilde{\pi}^{\star}}\left(\mathcal{M}^{\star}\right)-J^{\tilde{\pi}_{\operatorname {Bayes }}^{\star}}\left(\mathcal{M}^{\star}\right).
\end{equation}
Since the $\mathcal{M}^{\star}$ is unknown, this regret cannot be directly evaluated. Leveraging the theoretical properties of MPIL, the following theorem further establishes its connection to the true regret.

\begin{thm}
\label{upper-bound}
Under Assumption~\ref{asm:latent_structure} and bounded coverage, the regret satisfies
\begin{align}
\label{regret_sup_pi}
\operatorname{Regret}\left(\mathcal{M}^{\star}, \hat{\mathcal{D}}_{N}\right)
\leq 2 \mathcal{R}_{\max } \cdot  \sqrt{1-\operatorname{exp}\left(-\frac{C\mathcal{I}_{\hat{\mathcal{D}}_N}}{1-\gamma}\right)}.
\end{align}
Here, $C<\infty$ denotes a concentrability coefficient that bounds the distribution shift between the policy-induced state--action distribution and the offline data distribution, i.e., the relevant occupancy-density ratio is uniformly bounded by $C$. This assumption is closely related to the bounded importance-weight condition used in off-policy evaluation \citep{thomas2016dataefficient}.

\end{thm}

The proof is provided in Appendix~\ref{appendix:proof_thm1}.
Theorem~\ref{upper-bound} shows that MPIL upper-bounds the regret of the policy learned from the world model, with a smaller MPIL yielding a tighter performance guarantee. This provides support for using MPIL in offline world model refinement and selection.

\subsection{Expected Regret under MPIL Contraction}

Theorem~\ref{upper-bound} relates MPIL to policy regret. We further characterize its dependence on offline dataset size under an explicit MPIL-contraction assumption and a standard coverage condition \citep{ghosal2000convergence,wei2003predictive, uehara2022pessimistic,zhan2022offline}. Throughout this subsection, expectations are taken over the randomness of the offline dataset and any algorithmic randomness, when applicable.

\begin{thm}
\label{thm2}
Assume that the expected MPIL contracts at the parametric rate $\mathbb{E}_{\hat{\mathcal D}_N}[\mathcal I_{\hat{\mathcal D}_N}] \leq K/N$, where $K>0$ is independent of $N$. Assume further that the concentrability coefficient $C=C(\hat{\mathcal D}_N)$ in Theorem~\ref{upper-bound} is almost surely bounded, uniformly in $N$, by a deterministic constant $C_{\mathrm{cov}}<\infty$.
Then, the expected policy regret satisfies
\begin{equation}
\label{eq:expected_regret_rate}
\mathbb{E}_{\hat{\mathcal D}_N}
\left[
\operatorname{Regret}
\left(\mathcal M^{\star},\hat{\mathcal D}_N\right)
\right]
\leq
2\mathcal R_{\max}
\sqrt{
\frac{C_{\mathrm{cov}}K}
{(1-\gamma)N}
}.
\end{equation}
Posterior contraction together with bounded coverage yields an expected regret rate of $O(N^{-1/2})$.
\end{thm}

The proof is provided in Appendix~\ref{appendix:proof_thm2}. Theorem~\ref{thm2} separates two sources that determine the finite-data guarantee. The contraction assumption controls how quickly the posterior world model discrepancy decreases as more offline data become available, whereas $C_{\mathrm{cov}}$ captures the distribution mismatch between the offline data and the state-joint-action regions visited by $\tilde{\pi}^{\star}$ and $\tilde{\pi}_{\mathrm{Bayes}}^{\star}$. Therefore, increasing the dataset tightens the regret guarantee only when these policy-relevant regions remain sufficiently covered.

\subsection{Empirical Regret Approximation}
\label{sec:empirical_regret}
The regret in Eq.~(\ref{Regret}) cannot be evaluated directly because the true environment and its optimal policy are unknown. We therefore introduce an empirical approximation based on the return predicted by the learned posterior world model. Given a policy $\tilde{\pi}_{\operatorname{Bayes}}^{\star}$, its model-based return is estimated through finite-horizon rollouts as $\widehat{J}^{\tilde{\pi}_{\operatorname{Bayes}}^{\star}} (\hat{\mathcal{M}}_N) = \frac{1}{L} \sum_{\ell=1}^{L} \sum_{t=0}^{H-1} \hat{r}_{t}^{(\ell)}$,
where $\hat{\mathcal{M}}_N$ denotes the posterior world model learned from $\hat{\mathcal{D}}_N$, and $L$ is the number of model rollouts.
Let $J_{\mathrm{ref}}^{+}$ and $J_{\mathrm{ref}}^{-}$ denote fixed upper and lower reference returns, respectively. We define the empirical regret approximation as
\begin{equation}
    \widehat{\operatorname{Regret}} \left(\hat{\mathcal{M}}_N,\hat{\mathcal{D}}_N\right) := \frac{ J_{\mathrm{ref}}^{+} - \widehat{J}^{\tilde{\pi}_{\operatorname{Bayes}}^{\star}} (\hat{\mathcal{M}}_N) }{ J_{\mathrm{ref}}^{+}-J_{\mathrm{ref}}^{-} }.
    \label{eq:empirical_regret}
\end{equation}
This normalization provides a comparable measure of the predicted performance gap, with a smaller value indicating better policy performance. Although it does not replace the theoretical regret in Eq.~(\ref{Regret}), it provides an approximate method for evaluating policy quality entirely offline.

%% file: sections/07_Experiments.tex
\section{Experiments}

We adopt three classic MuJoCo continuous-control tasks from OpenAI Gym, including Halfcheetah, Hopper, and Walker2d, as the experimental environments. All experiments are conducted using the D4RL benchmark~\citep{fu2020d4rl} for offline RL. For each task, four types of offline datasets are considered, including random, medium, full-replay, and expert, resulting in a total of 12 task-dataset combinations. To investigate offline multi-agent cooperation under partial observability, we extend each original single-agent dataset to a multi-agent setting. 
Each agent has access only to its local observations, while the joint actions of all agents collectively determine the environment dynamics. 
To ensure a consistent information setting across different methods, the global state is excluded from the offline datasets and remains unavailable to all methods throughout training.

\subsection{Offline Policy Performance}
Policy performance is evaluated in terms of normalized average return and compared with behavior cloning (BC), OMAR~\citep{pan2022plan}, and MOMA-PPO~\citep{barde2024moma}. Once the world model satisfies the MPIL criterion, we fix the model posterior and generate imagined trajectories for cooperative multi-agent policy optimization under the CTDE paradigm. Table~\ref{tab:normalized_performance_comparison} reports the normalized returns across the three tasks and four dataset types, using the standard D4RL reference scores. Results are reported as mean $\pm$ standard deviation over 5 random seeds, with each final policy evaluated over 10 episodes. Implementation details and hyperparameter settings for all methods are provided in Appendix~\ref{app:details}.

\begin{table}[htbp]
\centering
\setlength{\tabcolsep}{1.5pt}
\caption{Normalized performance comparison on the multi-agent variants of the D4RL datasets.}
\label{tab:normalized_performance_comparison}
\small
\begin{tabular*}{\linewidth}{@{\extracolsep{\fill}}llcccc}
\toprule
Task & Dataset &
\makecell{BC} &
\makecell{OMAR} &
\makecell{MOMA-PPO} &
\makecell{Our Work}\\
\midrule

\multirow{4}{*}{Hopper}
& expert
& $\mathbf{37.40 \pm 4.97}$
& $1.82 \pm 0.01$
& $18.86 \pm 2.97$
& $29.16 \pm 4.48$ \\

& full-replay
& $14.73 \pm 2.30$
& $1.86 \pm 0.01$
& $31.30 \pm 2.98$
& $\mathbf{34.72 \pm 4.66}$ \\

& medium
& $21.79 \pm 4.65$
& $11.79 \pm 0.91$
& $13.92 \pm 3.13$
& $\mathbf{28.63 \pm 3.01}$ \\

& random
& $3.24 \pm 0.60$
& $9.49 \pm 0.06$
& $16.96 \pm 2.49$
& $\mathbf{31.16 \pm 0.66}$ \\

\midrule

\multirow{4}{*}{Halfcheetah}
& expert
& $5.59 \pm 3.65$
& $8.56 \pm 5.14$
& $2.61 \pm 0.76$
& $\mathbf{19.80 \pm 2.34}$ \\

& full-replay
& $25.91 \pm 5.94$
& $33.47 \pm 9.61$
& $\mathbf{50.13 \pm 0.64}$
& $43.32 \pm 3.51$ \\

& medium
& $44.62 \pm 9.61$
& $41.80 \pm 1.12$
& $\mathbf{50.35 \pm 0.53}$
& $46.09 \pm 3.40$ \\

& random
& $2.25 \pm 0.01$
& $16.93 \pm 0.74$
& $15.18 \pm 0.27$
& $\mathbf{24.38 \pm 0.31}$ \\

\midrule

\multirow{4}{*}{Walker2d}
& expert
& $\mathbf{106.02 \pm 2.61}$
& $1.79 \pm 0.25$
& $75.19 \pm 14.66$
& $85.35 \pm 11.44$ \\

& full-replay
& $22.45 \pm 6.57$
& $1.03 \pm 0.15$
& $26.52 \pm 7.28$
& $\mathbf{63.57 \pm 11.93}$ \\

& medium
& $20.28 \pm 7.81$
& $21.32 \pm 3.98$
& $7.96 \pm 0.79$
& $\mathbf{29.62 \pm 6.02}$ \\

& random
& $1.51 \pm 0.18$
& $1.77 \pm 0.57$
& $1.18 \pm 0.14$
& $\mathbf{2.22 \pm 0.77}$ \\

\bottomrule
\end{tabular*}

\end{table}
Overall, the proposed method demonstrates strong and consistent policy performance across different tasks and dataset types, achieving the highest normalized return in eight of the twelve experimental settings and ranking second in the remaining four, with particularly clear advantages on the random and full-replay datasets. In comparison, BC performs well on the expert datasets of Hopper and Walker2d but degrades substantially on the random and some full-replay datasets, reflecting its strong dependence on the quality of the offline behavior data. MOMA-PPO achieves the best performance on the Halfcheetah full-replay and medium datasets, but its performance is less consistent across the other tasks. OMAR also exhibits noticeable performance variation across environments and yields relatively low returns in most of the settings considered. 
These results indicate that the proposed method adapts well to variations in offline data quality and coverage, and can extract decision-relevant information from offline trajectories rather than relying primarily on high-quality behavior demonstrations.

\subsection{Validation of MPIL and Approximate Regret}

\noindent
\begin{minipage}[t]{0.55\textwidth}
    \vspace{0pt}
    To validate the proposed offline evaluation mechanism, we jointly analyze $\delta_{\mathrm{MPIL}}$ and approximate regret across different offline dataset sizes.
    As shown in Fig.~\ref{fig:regret}, both $\delta_{\mathrm{MPIL}}$ and approximate regret decrease substantially at the early stage with increasing offline data in the Walker2d environment. This indicates that, as the posterior world model becomes more accurate, the downstream policy learned from it also exhibits improved performance. Once the dataset size reaches approximately 0.5M, both metrics gradually stabilize, suggesting that the information provided by additional offline data yields only limited further improvement to either the world model or the learned policy. 
    Moreover, the approximate regret follows a trend consistent with the true regret and stabilizes at a similar data scale, supporting its validity as an offline proxy for policy performance. The consistent stabilization of $\delta_{\mathrm{MPIL}}$, approximate regret, and true regret further provides a practical indication of data sufficiency, allowing the data scale required for stable policy performance to be estimated without online interaction.
    
    
\end{minipage}
\hfill
\begin{minipage}[t]{0.40\textwidth}
    \vspace{0pt}
    \centering
    \includegraphics[ width=\linewidth, trim={40pt 30pt 0pt 130pt}, clip ]{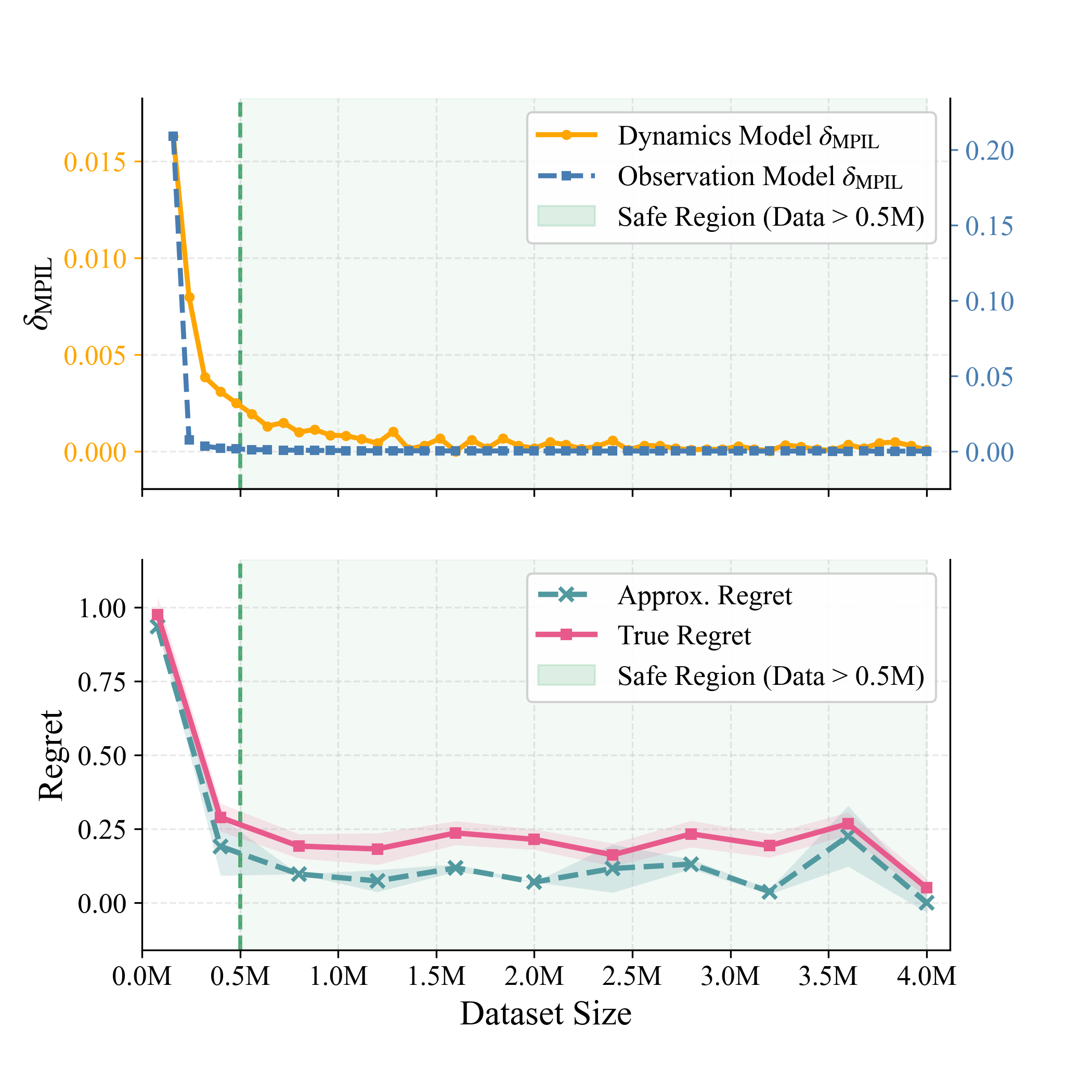}
    \captionsetup{skip=1 pt,hypcap=false}
    \captionof{figure}{Relationship between $\delta_{\mathrm{MPIL}}$, Approximate Regret, and True Regret under Different Offline Dataset Sizes on Walker2d. The True Regret Is Normalized Using the Standard D4RL Reference Scores.}
    \label{fig:regret}
\end{minipage}

\subsection{Effect of MPIL-Guided World Model Refinement}

To further examine the effect of MPIL-guided world model refinement, we compare the predictive performance of the world models and their corresponding downstream policies before and after refinement under the same data budget. 
For each of the four dataset types, the baseline model is trained directly on the full set of 0.5M samples, whereas the refined model is initialized with a subset and progressively updated with additional samples according to the MPIL criterion. Both models are evaluated on the corresponding validation set and on a mixed validation set constructed from all four dataset types. As an additional reference, we also train a separate world model on 0.5M samples from the mixed dataset and evaluate it using the same validation settings.
World model prediction performance is measured by the joint normalized root mean square error (JNRMSE), which jointly characterizes observation and reward prediction errors. Since simulator global states are excluded and not predicted, state prediction error is omitted.

As shown in Fig.~\ref{fig:walker2d_wm_error}, MPIL-guided refinement in the Walker2d environment generally reduces the prediction error of the world models trained on individual datasets.
The improvement is particularly evident on the mixed validation set, where the initial models exhibit substantially larger errors under the shifted data distribution. After refinement, their JNRMSE values are markedly reduced and are generally lower than those of the reference model trained directly on 0.5M mixed samples. 
These results indicate that MPIL-guided refinement can utilize a fixed amount of offline data more effectively, yielding more accurate and robust world model predictions across different data distributions.
The benefits of world model refinement are further reflected in policy learning. As shown in Fig.~\ref{fig:walker2d_policy}, across all four Walker2d datasets, policies trained with the refined world models achieve higher returns, with the improvement being particularly pronounced on the full-replay dataset. This indicates that improved world model predictions can provide more reliable imagined trajectories and thereby enhance offline policy optimization.

\begin{figure}[h]
    \centering
    \begin{subfigure}[b]{0.42\linewidth}
        \centering
        \includegraphics[height=5cm]{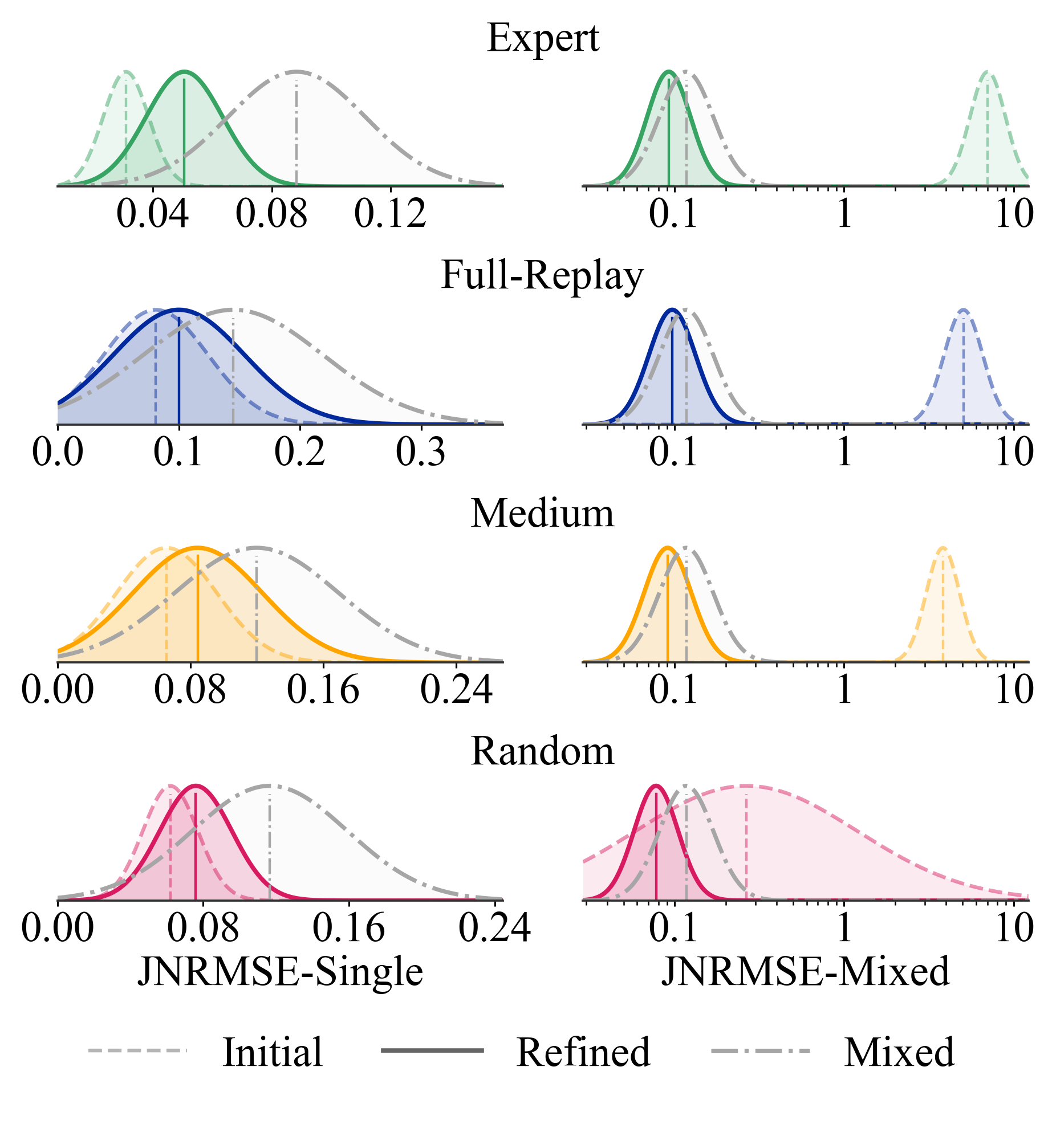}
        \caption{World model predictive performance}
        \label{fig:walker2d_wm_error}
    \end{subfigure}%
    \hfill%
    \begin{subfigure}[b]{0.57\linewidth}
        \centering
        \includegraphics[height=5cm]{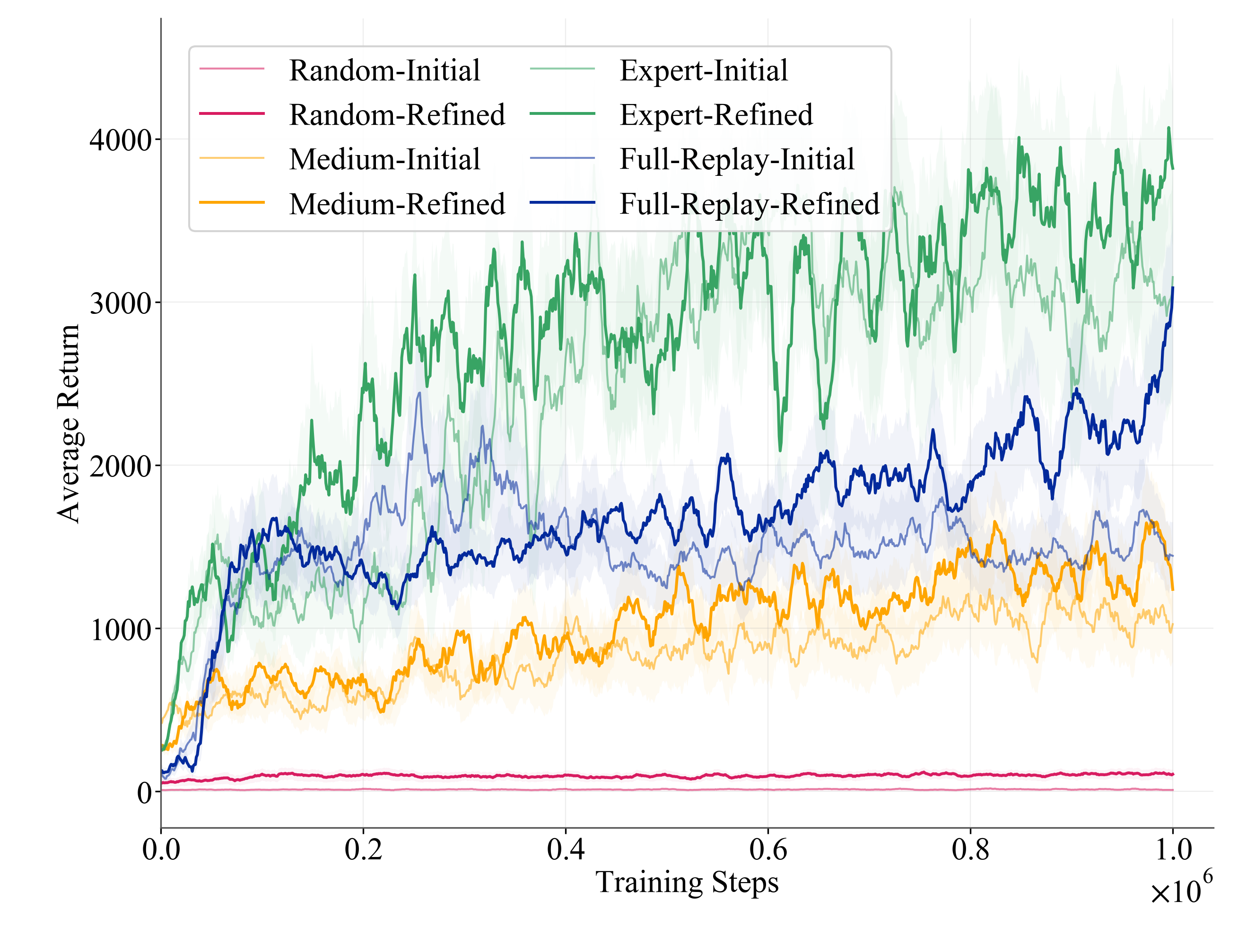}
        \caption{Policy performance}
        \label{fig:walker2d_policy}
    \end{subfigure}
    \caption{Effect of world model refinement on policy performance in Walker2d.}
    \label{fig:wm_refinement}
\end{figure}



%% file: sections/08_Conclusion.tex
\section{Conclusion}

This paper studies fully offline MARL under decentralized partial observability, where the global state is unavailable during training. We develop a Bayesian world model approach that infers latent states from local observation histories and models transition-reward and observation dynamics through posterior distributions. MPIL is further introduced to refinement of the posterior world model and enable fully offline evaluation, while posterior-consistent imagined trajectories are used to optimize cooperative multi-agent policies.
Theoretical analysis establishes a monotone upper bound on policy regret in terms of MPIL and further characterizes its dependence on the size and coverage of the offline dataset. 
Experiments on partially observable multi-agent MuJoCo tasks show that the proposed method achieves competitive policy performance across datasets of varying quality. MPIL decreases consistently with policy regret as offline information increases, while MPIL-guided refinement further improves world model prediction and downstream policy performance.

Future work will investigate zero-shot generalization to unseen agent configurations, interaction structures, and environment dynamics. We will also explore shared or conditioned world models that can reuse learned representations and dynamics across different task settings rather than retraining a new model from scratch. When model refinement is required, MPIL may further be used to identify unreliable regions and select the most informative offline samples for targeted updates, improving the data efficiency of world model refinement.

%% file: sections/09_Appendix.tex
\section{Appendix}
\subsection{Algorithm}
\begin{algorithm}[htbp]
\caption{MPIL-Guided Bayesian Offline MARL}
\label{alg:pil_marl}
\begin{algorithmic}[1]
\REQUIRE Raw offline data subsets $\{\mathcal{D}_m\}_{m=1}^{R}$, convergence tolerance $\epsilon_{\Delta}$, priors $P_{\Theta}$ and $P_{\tilde{\Omega}}$

\STATE Construct $\hat{\mathcal{D}}^{(1)}$ from $\mathcal{D}_1$ via latent-state inference
\STATE Learn the initial ensemble parameters $\Phi^{(1)}$ by minimizing $\mathcal{L}_{\mathrm{WM}}^{(1)}$
\STATE Initialize the terminating stage: $m^{\star}\leftarrow R$

\FOR{$m=1,\ldots,R-1$}
    \STATE Incorporate $\mathcal{D}_{m+1}$ to construct $\hat{\mathcal{D}}^{(m+1)}$
    \STATE Obtain $\Phi^{(m+1)}$ by warm-starting from $\Phi^{(m)}$ and minimizing $\mathcal{L}_{\mathrm{WM}}^{(m+1)}$
    \STATE Compute $\delta_{\mathrm{MPIL}}^{(m+1)}=\left|\mathcal{I}_{\hat{\mathcal{D}}_{\mathrm{val}}}(\Phi^{(m+1)})-\mathcal{I}_{\hat{\mathcal{D}}_{\mathrm{val}}}(\Phi^{(m)})\right|$
    \IF{$\delta_{\mathrm{MPIL}}^{(m+1)}\leq\epsilon_{\Delta}$}
        \STATE Record the terminating stage: $m^{\star}\leftarrow m+1$
        \STATE \textbf{break}
    \ENDIF
\ENDFOR

\STATE Construct the final empirical posterior:
$P_{\Theta,\tilde{\Omega}}^{(m^{\star})}(\cdot) \leftarrow \frac{1}{K}\sum_{k=1}^{K} \delta_{(\theta_k^{(m^{\star})},\tilde{\omega}_k^{(m^{\star})})}(\cdot)$

\STATE Sample a world model and an imagined trajectory: $(\theta,\tilde{\omega}) \sim P_{\Theta,\tilde{\Omega}}^{(m^{\star})}, \quad \check{h}_{\infty} \sim P_{\infty,\tilde{\pi}}(\theta,\tilde{\omega})$

\STATE Optimize the joint policy via CTDE:
$\tilde{\pi}_{\operatorname{Bayes}}^{\star} \leftarrow \underset{\tilde{\pi}}{\arg\max}\; J_{\mathrm{Bayes}}^{\tilde{\pi}} \left(P_{\Theta,\tilde{\Omega}}^{(m^{\star})}\right)$

\RETURN Bayes-optimal joint policy $\tilde{\pi}_{\operatorname{Bayes}}^{\star}$
\end{algorithmic}
\end{algorithm}

\subsection{Implementation Details}
\label{app:details}
\subsubsection{World Model Implementation}
\label{app:wm_implementation}
The latent-state inference module uses a recurrent encoder to summarize the joint observation--action history into a centralized latent state. At each time step, the joint observation and previous joint action are first projected through a linear layer with ReLU activation and then passed to a GRUCell. Given the inferred latent state and current joint action, an observation decoder predicts the next joint observation, while a shared transition--reward network predicts the next latent state and reward. The implementation details are summarized in Table~\ref{tab:latent_model_hyperparameters}.

\begin{table}[h]
    \centering
    \caption{Latent-state inference model hyperparameters.}
    \label{tab:latent_model_hyperparameters}
    \begin{tabular}{lc}
        \toprule
        \textbf{Hyperparameter} & \textbf{Value} \\
        \midrule
        Observation dimension & 17 \\
        Joint action dimension & 6 \\
        Latent state dimension & 64 \\
        GRU hidden dimension & 64 \\
        Encoder projection dimension & 256 \\
        Decoder hidden dimension & 256 \\
        Transition--reward hidden dimension & 256 \\
        Decoder MLP layers & 2 \\
        Transition--reward MLP layers & 2 \\
        Sequence length & 32 \\
        Input normalization & Yes \\
        Optimizer & Adam \\
        Pretraining learning rate & $3\times10^{-4}$ \\
        Fine-tuning learning rate & $1\times10^{-4}$ \\
        Activation function & ReLU; GRU sigmoid/tanh \\
        $\lambda\_{\mathrm{recon}}$ & 1.0 \\
        $\lambda\_{z}$ & 1.0 \\
        $\lambda\_{\mathrm{reward}}$ & 1.0 \\
        \bottomrule
    \end{tabular}
\end{table}

The inferred latent states are then used to construct the posterior world model described in Section~\ref{subsec:bayesian_wm_formulation}. The posterior over the transition--reward and observation-model parameters is approximated by an ensemble of probabilistic neural networks, with elite members selected according to holdout performance for posterior prediction and imagined rollouts. The ensemble is optimized using the empirical MPIL objective defined in Section~\ref{subsec:bayesian_wm_formulation}. The main implementation hyperparameters are summarized in Table~\ref{tab:wm_hyperparameters}.

\begin{table}[h]
    \centering
    \caption{Posterior world model implementation hyperparameters.}
    \label{tab:wm_hyperparameters}
    \begin{tabular}{lc}
        \toprule
        \textbf{Hyperparameter} & \textbf{Value} \\
        \midrule
        Ensemble size & 7 \\
        Number of elite models & 5 \\
        Hidden dimension & 300 \\
        Probabilistic output & Mean and variance \\
        Holdout ratio & 0.2 \\
        Optimizer & Adam\\
        Learning rate & $1\times10^{-3}$ \\
        Activation function & Swish; linear output; bounded log-variance \\
        \bottomrule
    \end{tabular}
\end{table}

\subsubsection{World Model Refinement}
\label{app:wm_refinement}
The refinement procedure follows Section~\ref{subsec:4.3} and uses the same world model ensemble architecture as~\ref{app:wm_implementation}. Each stage is warm-started from the preceding ensemble and optimized using the empirical MPIL objective in Section~\ref{subsec:bayesian_wm_formulation}, with holdout-based validation and early stopping. The main hyperparameters are summarized in Table~\ref{tab:refinement_hyperparameters}.

\begin{table}[h]
    \centering
    \caption{Hyperparameters for MPIL-based world model refinement.}
    \label{tab:refinement_hyperparameters}
    \begin{tabular}{lc}
        \toprule
        \textbf{Hyperparameter} & \textbf{Value} \\
        \midrule
        Maximum refinement rounds & 50 \\
        Epochs per round & 10 \\
        Batch size & 512 \\
        Holdout ratio & 0.2 \\
        Early-stopping patience & 5 \\
        Optimizer & Adam\\
        Learning rate & $1\times10^{-3}$ \\
        Activation function & Swish; linear output; bounded log-variance \\
        \bottomrule
    \end{tabular}
\end{table}

\subsubsection{Policy Learning with the Refined World Model}
\label{app:policy_learning}
Policy optimization follows the SAC-based CTDE formulation in Section~\ref{subsec:4.3}. Each decentralized actor is implemented as a two-layer MLP with a Gaussian policy output, while the centralized critic uses twin Q-networks with the same backbone structure. Imagined transitions generated from the refined world model ensemble are combined with offline data for actor and critic updates. The main policy-learning and rollout hyperparameters are summarized in Table~\ref{tab:policy_hyperparameters}.

\begin{table}[h]
    \centering
    \caption{Policy learning and model rollout hyperparameters.}
    \label{tab:policy_hyperparameters}
    \begin{tabular}{lc}
        \toprule
        \textbf{Hyperparameter} & \textbf{Value} \\
        \midrule
        Actor hidden dimensions & $[256,256]$ \\
        Critic hidden dimensions & $[256,256]$ \\
        Optimizer & Adam\\
        Learning rate & $5\times10^{-5}$ \\
        Activation function & ReLU; linear output; tanh squash \\
        Discount factor $\gamma$ & 0.99 \\
        Target update coefficient & 0.005 \\
        Batch size & 256 \\
        Total training steps & $1\times10^6$ \\
        Rollout length & 1 \\
        Rollout batch size & 50,000 \\
        Rollout frequency & 1000 steps \\
        Conservative loss weight & 5.0 \\
        Evaluation episodes & 10 \\
        \bottomrule
    \end{tabular}
\end{table}

\subsubsection{Baseline Implementation}
\label{app:baseline_implementation}

We compare our method with behavior cloning (BC), OMAR, and
MOMA-PPO~\citep{barde2024moma} under the same D4RL tasks and
multi-agent action decomposition. The main baseline hyperparameters are summarized in
Table~\ref{tab:baseline_hyperparameters}.

BC is implemented as deterministic behavior cloning with a
mean-squared action loss. OMAR follows the conservative critic
regularization and candidate-action-based actor rectification of
\citet{pan2022plan}, with a conservative critic weight of 2.0, a
behavior-cloning weight of 1.0, an OMAR objective weight of 0.5, a
candidate-action number of 8, and a critic-target clipping threshold
of 100.

MOMA-PPO follows the model-based offline MARL procedure of
\citet{barde2024moma}. A learned world model is used to generate
synthetic multi-agent trajectories initialized from the offline
dataset. At each rollout step, actions are sampled from the current
decentralized policies and the world model predicts the subsequent
transition and reward. Model uncertainty is used to penalize
unreliable synthetic transitions and to terminate rollouts when the
model becomes insufficiently reliable. The resulting synthetic
trajectories are then used to update the agents with MAPPO under
centralized training and decentralized execution.

\begin{table}[H]
    \centering
    \caption{Main hyperparameters of the baseline methods for Walker2d-random.}
    \label{tab:baseline_hyperparameters}
    \begin{tabular}{lccc}
        \toprule
        \textbf{Hyperparameter} & \textbf{BC} & \textbf{OMAR} & \textbf{MOMA-PPO} \\
        \midrule
        Training steps
        & $1\times10^6$
        & $1\times10^6$
        & $1\times10^6$ \\

        Batch size
        & 2048
        & 512
        & 256 \\

        Hidden dimension
        & 128
        & 300
        & $[256,256]$ \\

        Learning rate
        & $3\times10^{-5}$
        & $3\times10^{-5}$
        & $5\times10^{-5}$ \\

        Discount factor
        & --
        & 0.99
        & 0.99 \\

        Entropy coefficient
        & --
        & 0.02
        & 0.05 \\

        Gradient clipping
        & 5.0
        & 5.0
        & -- \\

        \bottomrule
    \end{tabular}
\end{table}

\subsubsection{Multi-Agent, Local-Observation, and Dataset Configuration}
\label{app:multi_agent_configuration}

Table~\ref{tab:multi_agent_configuration} gives the exact task decomposition. All indices are zero-based, and entries within each cell are ordered by agent ID. At time $t$, actor $i$ receives only $o_t^i=s_t[\mathcal O_i]$ and writes its output to the specified coordinates of the joint action $a_t$; the environment then executes the complete joint action.

\begin{table}[H]
    \centering
    \caption{Multi-agent action and observation decomposition. The task dimensions are shown as (observation, action).}
    \label{tab:multi_agent_configuration}
    \scriptsize
    \setlength{\tabcolsep}{3pt}
    \renewcommand{\arraystretch}{1.15}
    \resizebox{\textwidth}{!}{%
    \begin{tabular}{@{}p{0.13\textwidth}p{0.06\textwidth}p{0.16\textwidth}p{0.43\textwidth}p{0.18\textwidth}@{}}
        \toprule
        \textbf{Task (dims.)}
        & \makecell[l]{\textbf{No.}\\\textbf{agents}}
        & \makecell[l]{\textbf{Action dim.:}\\\textbf{indices}}
        & \textbf{Local observation indices $\mathcal O_i$}
        & \makecell[l]{\textbf{Shared/overlap}\\\textbf{indices}} \\
        \midrule
        HalfCheetah-v2\newline $(17,6)$
        & 3
        & \makecell[l]{$2:[0,1]$\\$2:[2,3]$\\$2:[4,5]$}
        & \makecell[l]{$\mathcal O_0=[0,1,2,3,8,9,10,11,12]$\\$\mathcal O_1=[0,1,4,5,8,9,10,13,14]$\\$\mathcal O_2=[0,1,6,7,8,9,10,15,16]$}
        & $[0,1,8,9,10]$ \\
        \addlinespace
        Hopper-v2\newline $(11,3)$
        & 3
        & \makecell[l]{$1:[0]$\\$1:[1]$\\$1:[2]$}
        & \makecell[l]{$\mathcal O_0=[0,1,2,5,6,7,8]$\\$\mathcal O_1=[0,1,3,5,6,7,9]$\\$\mathcal O_2=[0,1,4,5,6,7,10]$}
        & $[0,1,5,6,7]$ \\
        \addlinespace
        Walker2d-v2\newline $(17,6)$
        & 2
        & \makecell[l]{$3:[0,1,2]$\\$3:[3,4,5]$}
        & \makecell[l]{$\mathcal O_0=[0,1,2,3,4,8,9,10,11,12,13]$\\$\mathcal O_1=[0,1,5,6,7,8,9,10,14,15,16]$}
        & $[0,1,8,9,10]$ \\
        \bottomrule
    \end{tabular}
    }
\end{table}

Thus, partial observability is imposed on the decentralized actor inputs: the union of local observations reconstructs the original D4RL observation, whereas no single actor observes it in full. 

\paragraph{Dataset IDs and source.}
We use the D4RL Gym-MuJoCo \texttt{v2} datasets~\citep{fu2020d4rl} listed in Table~\ref{tab:d4rl_dataset_ids}.

\begin{table}[H]
    \centering
    \caption{D4RL Gym-MuJoCo \texttt{v2} dataset IDs used in our experiments.}
    \label{tab:d4rl_dataset_ids}
    \scriptsize
    \setlength{\tabcolsep}{4pt}
    \renewcommand{\arraystretch}{1.12}
    \resizebox{\textwidth}{!}{%
    \begin{tabular}{@{}lcccc@{}}
        \toprule
        \textbf{Task} & \textbf{Random} & \textbf{Medium} & \textbf{Expert} & \textbf{Full-replay} \\
        \midrule
        HalfCheetah
        & \texttt{halfcheetah-random-v2}
        & \texttt{halfcheetah-medium-v2}
        & \texttt{halfcheetah-expert-v2}
        & \texttt{halfcheetah-full-replay-v2} \\
        Hopper
        & \texttt{hopper-random-v2}
        & \texttt{hopper-medium-v2}
        & \texttt{hopper-expert-v2}
        & \texttt{hopper-full-replay-v2} \\
        Walker2d
        & \texttt{walker2d-random-v2}
        & \texttt{walker2d-medium-v2}
        & \texttt{walker2d-expert-v2}
        & \texttt{walker2d-full-replay-v2} \\
        \bottomrule
    \end{tabular}
    }
\end{table}

The \texttt{full-replay-v2} datasets are the complete replay buffers distributed by D4RL, rather than renamed \texttt{medium-replay} datasets or mixtures constructed by us. They contain neither model rollouts generated by our method nor trajectories resampled from the evaluation environments.

\subsection{Proof}

\begin{asm}[Latent sufficiency and factorization]
\label{asm:latent_structure}
For the theoretical analysis, we consider the realizable idealization in which $\hat{s}_t$ is predictively sufficient---future environment variables depend on the joint history only through $\hat{s}_t$. Let $x=(\hat{s},\tilde a)$, $u=(\hat{s}',r)$, and $v=\tilde z'$. The true and learned one-step kernels used in the analysis are
\begin{align}
K^\star(\mathrm du,\mathrm dv\mid x)
&:=P_{T,R}^\star(\mathrm du\mid x)
P_{\tilde O}^\star(\mathrm dv\mid x),
\nonumber\\
K_{\theta,\tilde\omega}(\mathrm du,\mathrm dv\mid x)
&:=P_{T,R}(\mathrm du\mid x,\theta)
P_{\tilde O}(\mathrm dv\mid x,\tilde\omega).
\end{align}
For every $\mu\in\{\nu_{\hat{\mathcal D}_N},\rho^{\tilde\pi^\star},\rho^{\tilde\pi_{\mathrm{Bayes}}^\star}\}$, we further assume that $K^\star(\cdot\mid x)\ll K_{\theta,\tilde\omega}(\cdot\mid x)$ for $\mu$-almost every $x$ and posterior-almost every $(\theta,\tilde\omega)$, and that
\begin{equation}
\mathbb E_{x\sim\mu}\mathbb E_{(\theta,\tilde\omega)\sim P_{\Theta,\tilde\Omega}(\cdot\mid\hat{\mathcal D}_N)}
\!\left[D_{\mathrm{KL}}\!\left(K^\star(\cdot\mid x)\middle\|K_{\theta,\tilde\omega}(\cdot\mid x)\right)\right]<\infty.
\end{equation}
These conditions define the scope of the theoretical analysis and are not implied by the recurrent encoder. Equation~(1) is a practical deterministic history encoder; Theorem~\ref{upper-bound} does not quantify deviations from this idealized setting.
\end{asm}

\subsubsection{Proof of Proposition~\ref{proposition_nll_mpil}}
\label{appendix:proof_prp1}

\begin{proof}
For notational compactness, let $x=(\hat{s},\tilde{a})$, let $y_{T,R}=(s',r)$, and let $y_{\tilde O}=\tilde{o}'$. For any fixed $x$ and posterior parameter realization $\theta$, the definition of the KL divergence gives
\begin{align}
&D_{\mathrm{KL}}\!\left( P_{T,R}^{\star}(\cdot\mid x) \middle\| P_{T,R}(\cdot\mid x,\theta) \right)
\nonumber\\
&\quad= \mathbb{E}_{y_{T,R}\sim P_{T,R}^{\star}(\cdot\mid x)} \left[ \log P_{T,R}^{\star}(y_{T,R}\mid x) -\log P_{T,R}(y_{T,R}\mid x,\theta) \right]
\nonumber\\
&\quad= \mathbb{E}_{y_{T,R}\sim P_{T,R}^{\star}(\cdot\mid x)} \left[ -\log P_{T,R}(y_{T,R}\mid x,\theta) \right] -H\!\left(P_{T,R}^{\star}(\cdot\mid x)\right).
\label{eq:appendix_kl_ce_tr}
\end{align}
Likewise, for any fixed $x$ and $\tilde{\omega}$,
\begin{align}
&D_{\mathrm{KL}}\!\left( P_{\tilde O}^{\star}(\cdot\mid x) \middle\| P_{\tilde O}(\cdot\mid x,\tilde{\omega}) \right)
\nonumber\\
&\quad= \mathbb{E}_{y_{\tilde O}\sim P_{\tilde O}^{\star}(\cdot\mid x)} \left[ -\log P_{\tilde O}(y_{\tilde O}\mid x,\tilde{\omega}) \right] -H\!\left(P_{\tilde O}^{\star}(\cdot\mid x)\right).
\label{eq:appendix_kl_ce_obs}
\end{align}
These identities hold for probability mass functions and, with $H$ interpreted as differential entropy, for probability densities, whenever the displayed quantities are well defined.

Substituting Eqs.~(\ref{eq:appendix_kl_ce_tr}) and (\ref{eq:appendix_kl_ce_obs}) into the definition of MPIL in Eq.~(\ref{PIL}), and then taking expectations with respect to the fixed offline data distribution and the posterior, yields
\begin{align}
\mathcal{I}_{\hat{\mathcal D}_N} ={}
& \mathbb{E}_{x\sim\nu_{\hat{\mathcal D}_N}} \mathbb{E}_{(\theta,\tilde{\omega})\sim P_{\Theta,\tilde{\Omega}}(\cdot\mid\hat{\mathcal D}_N)}
\Big[ \mathbb{E}_{y_{T,R}\sim P_{T,R}^{\star}(\cdot\mid x)} [-\log P_{T,R}(y_{T,R}\mid x,\theta)]
\nonumber\\
&\hspace{35mm}+ \mathbb{E}_{y_{\tilde O}\sim P_{\tilde O}^{\star}(\cdot\mid x)} [-\log P_{\tilde O}(y_{\tilde O}\mid x,\tilde{\omega})] \Big]
\nonumber\\
&- \mathbb{E}_{x\sim\nu_{\hat{\mathcal D}_N}} \Big[ H\!\left(P_{T,R}^{\star}(\cdot\mid x)\right) +H\!\left(P_{\tilde O}^{\star}(\cdot\mid x)\right) \Big].
\label{eq:appendix_mpil_nll_expansion}
\end{align}
The first term on the right-hand side of Eq.~(\ref{eq:appendix_mpil_nll_expansion}) is exactly the posterior-averaged NLL $\mathcal{L}_{\mathrm{NLL}}$ defined in Eq.~(\ref{ideal_NLL}), while the second term is $\mathcal{H}^{\star}_{\hat{\mathcal D}_N}$. Consequently,
\begin{equation}
\mathcal{I}_{\hat{\mathcal D}_N} = \mathcal{L}_{\mathrm{NLL}} - \mathcal{H}^{\star}_{\hat{\mathcal D}_N},
\end{equation}
which proves Eq.~(\ref{MPIL_NLL}). Because the true-environment entropy term depends only on $P_{T,R}^{\star}$, $P_{\tilde O}^{\star}$, and the fixed distribution $\nu_{\hat{\mathcal D}_N}$, it is independent of the learned world model parameters. Hence differences in posterior-averaged NLL between world models evaluated under the same offline data distribution equal the corresponding differences in MPIL.
\end{proof}

\subsubsection{Proof of Theorem~\ref{upper-bound}}
\label{appendix:proof_thm1}

\begin{asm}
\label{asm:bounded_coverage}
For each relevant policy $\tilde{\pi}\in \{\tilde{\pi}^{\star},\tilde{\pi}_{\mathrm{Bayes}}^{\star}\}$, assume that the policy-induced state-action distribution $\rho^{\tilde{\pi}}$ is absolutely continuous with respect to the offline data distribution $\nu_{\hat{\mathcal D}_N}$ and that
\begin{align}
\label{eq:appendix_concentrability_coefficient}
C :={}& \max_{\tilde{\pi}\in \{\tilde{\pi}^{\star}, \tilde{\pi}_{\mathrm{Bayes}}^{\star}\}} \operatorname*{ess\,sup}_{(\hat{s},\tilde a)} \frac{ \rho^{\tilde{\pi}}(\hat{s},\tilde a) }{ \nu_{\hat{\mathcal D}_N}(\hat{s},\tilde a) } < \infty,
\end{align}
where the essential supremum is taken with respect to $\nu_{\hat{\mathcal D}_N}$.
\end{asm}

Assumption~\ref{asm:bounded_coverage} requires the offline data distribution to cover the state--action regions visited by the two policies relevant to the regret comparison, but does not require uniform coverage over all possible policies.

This coverage condition is closely related to the bounded importance-weight assumption commonly adopted in off-policy evaluation \citep{thomas2016dataefficient}: both require the target policies to remain within the support of the offline data-generating distribution. We use it here only as an explicit conditional assumption and do not claim that the policy-learning algorithm enforces it automatically.

\begin{proof}
For a joint policy $\tilde{\pi}$, write its posterior-predictive value as
\begin{equation}
\label{eq:appendix_bayesian_value}
J_{B}^{\tilde{\pi}} := J_{\mathrm{Bayes}}^{\tilde{\pi}} \left(P_{\Theta,\tilde{\Omega}}(\cdot\mid\hat{\mathcal D}_N)\right).
\end{equation}
Let $\tilde{\pi}^{\star}$ maximize the return in the true environment $\mathcal M^{\star}$, and let $\tilde{\pi}_{\mathrm{Bayes}}^{\star}$ maximize $J_B^{\tilde{\pi}}$. Adding and subtracting the two corresponding posterior-predictive values in Eq.~(\ref{Regret}) gives
\begin{align}
&\operatorname{Regret} \left(\mathcal M^{\star},\hat{\mathcal D}_N\right)
\nonumber\\
&={} J^{\tilde{\pi}^{\star}}(\mathcal M^{\star}) -J_B^{\tilde{\pi}^{\star}} +J_B^{\tilde{\pi}^{\star}} -J_B^{\tilde{\pi}_{\mathrm{Bayes}}^{\star}} +J_B^{\tilde{\pi}_{\mathrm{Bayes}}^{\star}} -J^{\tilde{\pi}_{\mathrm{Bayes}}^{\star}}(\mathcal M^{\star})
\nonumber\\
&\leq \left| J^{\tilde{\pi}^{\star}}(\mathcal M^{\star}) -J_B^{\tilde{\pi}^{\star}} \right| + \left| J_B^{\tilde{\pi}_{\mathrm{Bayes}}^{\star}} -J^{\tilde{\pi}_{\mathrm{Bayes}}^{\star}}(\mathcal M^{\star}) \right|,
\label{eq:appendix_two_policy_regret}
\end{align}
where the inequality follows from $J_B^{\tilde{\pi}^{\star}} -J_B^{\tilde{\pi}_{\mathrm{Bayes}}^{\star}}\leq0$. Thus, it suffices to control the return discrepancy for these two policies rather than uniformly over all policies.

We next establish such a bound for an arbitrary fixed joint policy $\tilde{\pi}$. Define its posterior-predictive trajectory-prefix law by
\begin{equation}
\label{eq:appendix_posterior_predictive_prefix}
\widehat P_{t+1,\tilde{\pi}}(\hat{\mathcal D}_N) := \mathbb E_{(\theta,\tilde{\omega})\sim P_{\Theta,\tilde{\Omega}}(\cdot\mid\hat{\mathcal D}_N)} \left[P_{t+1,\tilde{\pi}}(\theta,\tilde{\omega})\right].
\end{equation}
We assume that the true and posterior-predictive processes have the same initial latent-state and initial-observation distribution; an
additional initial-distribution divergence term is required otherwise. Since rewards lie in $[r_{\min},r_{\max}]$, the variational characterization of total variation implies
\begin{align}
&\left| J^{\tilde{\pi}}(\mathcal M^{\star})-J_B^{\tilde{\pi}} \right|
\nonumber\\
&\leq (r_{\max}-r_{\min}) \sum_{t=0}^{\infty}\gamma^t D_{\mathrm{TV}}\!\left( P_{t+1,\tilde{\pi}}(\mathcal M^{\star}), \widehat P_{t+1,\tilde{\pi}}(\hat{\mathcal D}_N) \right)
\nonumber\\
&= \mathcal R_{\max} \mathbb E_{t\sim\mathcal G(\gamma)} \left[ D_{\mathrm{TV}}\!\left( P_{t+1,\tilde{\pi}}(\mathcal M^{\star}), \widehat P_{t+1,\tilde{\pi}}(\hat{\mathcal D}_N) \right) \right],
\label{eq:appendix_fixed_policy_tv}
\end{align}
where $\mathcal R_{\max}=(r_{\max}-r_{\min})/(1-\gamma)$ and $\Pr_{\mathcal G(\gamma)}(T=t)=(1-\gamma)\gamma^t$.

Let $g(x):=\sqrt{1-\exp(-x)}$. The Bretagnolle--Huber inequality, convexity of KL divergence in its second argument, and monotonicity and concavity of $g$ yield
\begin{align}
&\mathbb E_{t\sim\mathcal G(\gamma)} \left[ D_{\mathrm{TV}}\!\left( P_{t+1,\tilde{\pi}}(\mathcal M^{\star}), \widehat P_{t+1,\tilde{\pi}}(\hat{\mathcal D}_N) \right) \right]
\nonumber\\
&\leq g\!\Bigg( \mathbb E_{t\sim\mathcal G(\gamma)} \mathbb E_{(\theta,\tilde{\omega})\sim P_{\Theta,\tilde{\Omega}}(\cdot\mid\hat{\mathcal D}_N)} \left[ D_{\mathrm{KL}}\!\left( P_{t+1,\tilde{\pi}}(\mathcal M^{\star}) \middle\| P_{t+1,\tilde{\pi}}(\theta,\tilde{\omega}) \right) \right] \Bigg).
\label{eq:appendix_tv_to_kl}
\end{align}

For $x=(\hat{s},\tilde a)$, define the posterior-averaged one-step model discrepancy
\begin{align}
\label{eq:appendix_pointwise_discrepancy}
d_{\hat{\mathcal D}_N}(x) :={}& \mathbb E_{(\theta,\tilde{\omega})\sim P_{\Theta,\tilde{\Omega}}(\cdot\mid\hat{\mathcal D}_N)} \Big[ D_{\mathrm{KL}}\!\left( P_{T,R}^{\star}(\cdot\mid x) \middle\| P_{T,R}(\cdot\mid x,\theta) \right)
\nonumber\\
&\hspace{32mm}+ D_{\mathrm{KL}}\!\left( P_{\tilde O}^{\star}(\cdot\mid x) \middle\| P_{\tilde O}(\cdot\mid x,\tilde{\omega}) \right) \Big].
\end{align}
Under Assumption~\ref{asm:latent_structure}, the KL identity for product measures gives, for each posterior-supported $(\theta,\tilde\omega)$,
\begin{align}
&D_{\mathrm{KL}}\!\left(
K^\star(\cdot\mid x)
\middle\|
K_{\theta,\tilde\omega}(\cdot\mid x)
\right)
\nonumber\\
&\quad=
D_{\mathrm{KL}}\!\left(
P_{T,R}^\star(\cdot\mid x)
\middle\|
P_{T,R}(\cdot\mid x,\theta)
\right)
+D_{\mathrm{KL}}\!\left(
P_{\tilde O}^\star(\cdot\mid x)
\middle\|
P_{\tilde O}(\cdot\mid x,\tilde\omega)
\right).
\label{eq:appendix_product_kernel_kl}
\end{align}
Thus, $d_{\hat{\mathcal D}_N}(x)$ is exactly the posterior-averaged KL divergence between the true and learned one-step kernels; no mutual-information term appears because both kernels factorize in Assumption~\ref{asm:latent_structure}.

For a fixed joint policy $\tilde\pi$, the true and learned trajectory-prefix laws factorize as
\begin{align}
P_{t+1,\tilde\pi}(\mathcal M^\star)
&=P_0\prod_{j=0}^{t}
\tilde\pi(\tilde a_j\mid\tilde h_j)
K^\star(u_j,v_j\mid x_j),
\nonumber\\
P_{t+1,\tilde\pi}(\theta,\tilde\omega)
&=P_0\prod_{j=0}^{t}
\tilde\pi(\tilde a_j\mid\tilde h_j)
K_{\theta,\tilde\omega}(u_j,v_j\mid x_j),
\label{eq:appendix_trajectory_kernel_factorization}
\end{align}
where $u_j=(\hat s_{j+1},r_j)$, $v_j=\tilde z_{j+1}$, and the common $P_0$ is the initial latent-state and observation law specified above. The common initial and policy factors therefore cancel in the likelihood ratio. Applying the KL chain rule and then Eq.~(\ref{eq:appendix_product_kernel_kl}) gives
\begin{align}
&\mathbb E_{(\theta,\tilde{\omega})\sim P_{\Theta,\tilde{\Omega}}(\cdot\mid\hat{\mathcal D}_N)} \left[ D_{\mathrm{KL}}\!\left(
P_{t+1,\tilde{\pi}}(\mathcal M^{\star}) \middle\| P_{t+1,\tilde{\pi}}(\theta,\tilde{\omega}) \right) \right]
\nonumber\\
&\qquad= \sum_{j=0}^{t} \mathbb E_{x_j\sim P_j^{\tilde{\pi}}(\mathcal M^{\star})} \left[d_{\hat{\mathcal D}_N}(x_j)\right].
\label{eq:appendix_history_kl_chain}
\end{align}
Define the normalized policy-induced state--action distribution by
\begin{equation}
\label{eq:appendix_policy_occupancy}
\rho^{\tilde{\pi}}
:=
\mathbb E_{t\sim\mathcal{AG}(\gamma)}
\mathbb E_{j\sim\mathcal U_t}
\left[P_j^{\tilde{\pi}}(\mathcal M^{\star})\right],
\end{equation}
where $\mathcal U_t$ is uniform on $\{0,\ldots,t\}$ and
$\Pr_{\mathcal{AG}(\gamma)}(T=t)
=(1-\gamma)^2(t+1)\gamma^t$. Combining
Eqs.~(\ref{eq:appendix_history_kl_chain}) and
(\ref{eq:appendix_policy_occupancy}), we obtain
\begin{align}
&\mathbb E_{t\sim\mathcal G(\gamma)}
\mathbb E_{(\theta,\tilde{\omega})\sim
P_{\Theta,\tilde{\Omega}}(\cdot\mid\hat{\mathcal D}_N)}
\left[
D_{\mathrm{KL}}\!\left(
P_{t+1,\tilde{\pi}}(\mathcal M^{\star})
\middle\|
P_{t+1,\tilde{\pi}}(\theta,\tilde{\omega})
\right)
\right]
\nonumber\\
&\qquad=
\frac{1}{1-\gamma}
\mathbb E_{x\sim\rho^{\tilde{\pi}}}
\left[d_{\hat{\mathcal D}_N}(x)\right].
\label{eq:appendix_discounted_kl_identity}
\end{align}
Let
\begin{equation}
\label{eq:appendix_policy_information_loss}
\mathcal I^{\tilde{\pi}}
:=
\mathbb E_{x\sim\rho^{\tilde{\pi}}}
\left[d_{\hat{\mathcal D}_N}(x)\right]
\end{equation}
denote the policy-dependent information loss used only as an
intermediate quantity in the proof. Equations~(\ref{eq:appendix_fixed_policy_tv})--
(\ref{eq:appendix_discounted_kl_identity}) imply
\begin{equation}
\label{eq:appendix_fixed_policy_information_bound}
\left|
J^{\tilde{\pi}}(\mathcal M^{\star})-J_B^{\tilde{\pi}}
\right|
\leq
\mathcal R_{\max}
\sqrt{
1-\exp\!\left(
-\frac{\mathcal I^{\tilde{\pi}}}{1-\gamma}
\right)}.
\end{equation}

It remains to connect the policy-dependent quantity in
Eq.~(\ref{eq:appendix_policy_information_loss}) to the
dataset-weighted MPIL. Under the bounded-coverage assumption above,
$d_{\hat{\mathcal D}_N}(x)\geq0$, so a change of measure gives, for
either relevant policy,
\begin{align}
\mathcal I^{\tilde{\pi}}
&=
\int
d_{\hat{\mathcal D}_N}(x)
\frac{d\rho^{\tilde{\pi}}}{d\nu_{\hat{\mathcal D}_N}}(x)
\,d\nu_{\hat{\mathcal D}_N}(x)
\nonumber\\
&\leq
C
\int d_{\hat{\mathcal D}_N}(x)
\,d\nu_{\hat{\mathcal D}_N}(x)
=
C\mathcal I_{\hat{\mathcal D}_N},
\label{eq:appendix_coverage_bridge}
\end{align}
where the last equality is the dataset-weighted MPIL definition in
Eq.~(\ref{PIL}). Substituting Eq.~(\ref{eq:appendix_coverage_bridge})
into Eq.~(\ref{eq:appendix_fixed_policy_information_bound}) for the
two policies in Eq.~(\ref{eq:appendix_two_policy_regret}), and using
the monotonicity of $g$, yields
\begin{align}
\operatorname{Regret}
\left(\mathcal M^{\star},\hat{\mathcal D}_N\right)
&\leq
2\mathcal R_{\max}
\sqrt{
1-\exp\!\left(
-\frac{C\mathcal I_{\hat{\mathcal D}_N}}{1-\gamma}
\right)}.\qedhere
\end{align}
\end{proof}

\subsubsection{Proof of Theorem~\ref{thm2}}
\label{appendix:proof_thm2}

\begin{proof}
Since $1-\exp(-x)\leq x$ for every $x\geq0$, Theorem~\ref{upper-bound} and the uniform bound $C\leq C_{\mathrm{cov}}$ imply
\begin{equation}
\label{eq:appendix_regret_sqrt_mpil}
\operatorname{Regret} \left(\mathcal M^{\star},\hat{\mathcal D}_N\right) \leq 2\mathcal R_{\max} \sqrt{ \frac{C_{\mathrm{cov}}\mathcal I_{\hat{\mathcal D}_N}} {1-\gamma} }.
\end{equation}
Taking expectations over the offline dataset and any algorithmic randomness, when applicable, and applying Jensen's inequality to the concave square-root function gives
\begin{align}
\mathbb E_{\hat{\mathcal D}_N} \left[ \operatorname{Regret} \left(\mathcal M^{\star},\hat{\mathcal D}_N\right)
\right]
&\leq \frac{2\mathcal R_{\max}\sqrt{C_{\mathrm{cov}}}} {\sqrt{1-\gamma}} \mathbb E_{\hat{\mathcal D}_N} \left[\sqrt{\mathcal I_{\hat{\mathcal D}_N}}\right]
\nonumber\\
&\leq \frac{2\mathcal R_{\max}\sqrt{C_{\mathrm{cov}}}} {\sqrt{1-\gamma}} \sqrt{ \mathbb E_{\hat{\mathcal D}_N} \left[\mathcal I_{\hat{\mathcal D}_N}\right] }.
\label{eq:appendix_expected_regret_mpil}
\end{align}
Finally, substituting the assumed contraction rate $\mathbb E_{\hat{\mathcal D}_N}[\mathcal I_{\hat{\mathcal D}_N}] \leq K/N$ into Eq.~(\ref{eq:appendix_expected_regret_mpil}) yields
\begin{equation}
\mathbb E_{\hat{\mathcal D}_N} \left[ \operatorname{Regret} \left(\mathcal M^{\star},\hat{\mathcal D}_N\right)
\right] \leq 2\mathcal R_{\max} \sqrt{ \frac{C_{\mathrm{cov}}K}{(1-\gamma)N} },
\end{equation}
which is Eq.~(\ref{eq:expected_regret_rate}) and completes the proof.
\end{proof}

\clearpage
\subsection{World Model Predictive Performance}
\label{app:wm}

We evaluate next-observation and reward predictions on a held-out test set. Here, the 17/11/17 prediction dimensions correspond to the original D4RL observation vectors for HalfCheetah, Hopper, and Walker2d, respectively, rather than simulator global states or inferred latent states. To further examine whether a reduction in MPIL corresponds to more accurate modeling of the environment dynamics, we compare the next-observation and reward predictions of the initial world model with those of the world model optimized under MPIL guidance on the Walker2d task. As shown in Figure~\ref{prediction}, the dashed line $y=x$ represents ideal predictions, where the model outputs exactly match the ground-truth. For next-observation prediction, although the initial world model generally captures the trends across different observation dimensions, noticeable prediction dispersion remains in some regions. After the MPIL-guided posterior updates, the predictions across observation dimensions become more concentrated around $y=x$, with substantially reduced dispersion. A similar trend is observed for reward prediction. These results are consistent with the decreasing trend of MPIL during training, indicating that lower MPIL reflects improved characterization of next-observation and reward dynamics by the world model, thereby providing a more reliable model foundation for subsequent policy learning based on imagined trajectories.

\begin{figure}[htbp]
    \centering
    \includegraphics[width=1\textwidth]{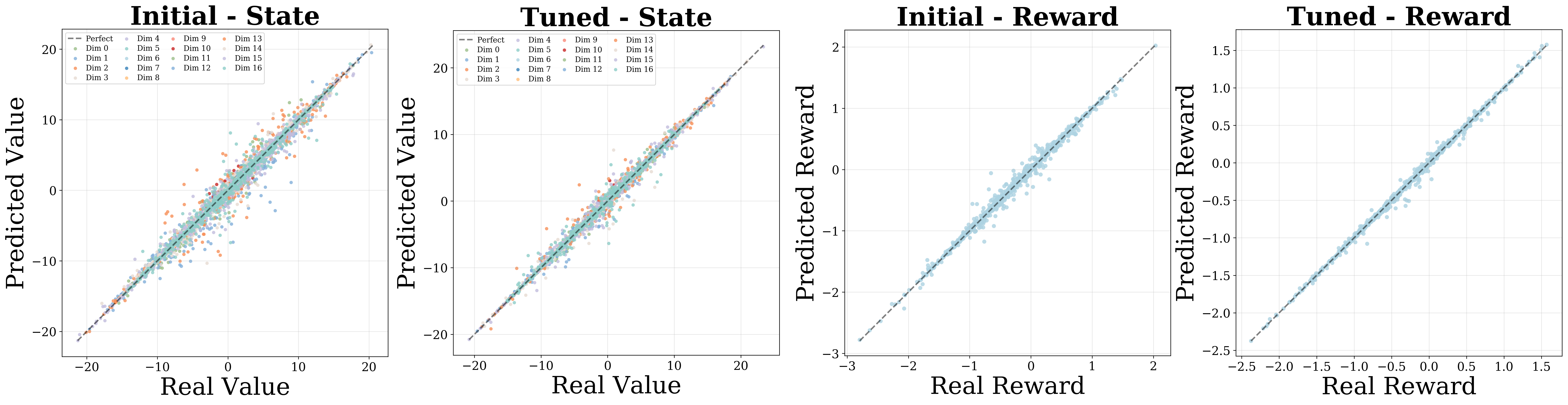}
    \caption{Comparison of World Model Predictions Before and After MPIL-Guided Optimization.}
    \label{prediction}
\end{figure}

\begin{table*}[htbp]
\centering
\caption{Comparison of prediction performance before and after refinement on different MuJoCo tasks.}
\label{tab:world_model_comparison}
\resizebox{\textwidth}{!}{
\begin{tabular}{llcccccc}
\toprule
Task & Dimension & MAE\_Initial& MAE\_Tuned& MAE\_Improvement (\%) & $R^2$\_Inital& $R^2$\_Tuned& $R^2$\_Improvement (\%) \\
\midrule

Halfcheetah & 0      & 0.0029 & 0.0012 & 58.93 & 0.9994 & 0.9999 & 0.05 \\
            & 1      & 0.0054 & 0.0022 & 58.69 & 0.9999 & 1.0000 & 0.00 \\
            & 2      & 0.0076 & 0.0030 & 59.99 & 0.9978 & 0.9996 & 0.18 \\
            & 3      & 0.0088 & 0.0033 & 62.69 & 0.9967 & 0.9995 & 0.28 \\
            & 4      & 0.0101 & 0.0040 & 60.54 & 0.9945 & 0.9991 & 0.45 \\
            & 5      & 0.0100 & 0.0045 & 55.14 & 0.9973 & 0.9993 & 0.20 \\
            & 6      & 0.0093 & 0.0043 & 54.13 & 0.9972 & 0.9992 & 0.20 \\
            & 7      & 0.0099 & 0.0053 & 46.44 & 0.9959 & 0.9985 & 0.26 \\
            & 8      & 0.0641 & 0.0313 & 51.20 & 0.9778 & 0.9928 & 1.53 \\
            & 9      & 0.1036 & 0.0499 & 51.86 & 0.9474 & 0.9793 & 3.37 \\
            & 10     & 0.2063 & 0.1000 & 51.53 & 0.9483 & 0.9807 & 3.41 \\
            & 11     & 0.2736 & 0.1337 & 51.14 & 0.9918 & 0.9973 & 0.55 \\
            & 12     & 0.4100 & 0.1884 & 54.04 & 0.9832 & 0.9956 & 1.26 \\
            & 13     & 0.3691 & 0.1684 & 54.38 & 0.9836 & 0.9957 & 1.23 \\
            & 14     & 0.3865 & 0.2009 & 48.03 & 0.9871 & 0.9951 & 0.81 \\
            & 15     & 0.4084 & 0.2355 & 42.34 & 0.9873 & 0.9932 & 0.60 \\
            & 16     & 0.3936 & 0.2276 & 42.18 & 0.9803 & 0.9904 & 1.03 \\
            & Reward & 0.0380 & 0.0180 & 52.72 & 0.9903 & 0.9973 & 0.71 \\

Hopper      & 0      & 0.0054 
& 0.0006
& 88.17 & 0.1732 & 0.9791 & 465.28 \\
            & 1      & 0.0079 
& 0.0007
& 90.73 & 0.8290 & 0.9918 & 19.64 \\
            & 2      & 0.0084 
& 0.0015
& 82.22 & 0.8145 & 0.9720 & 19.34 \\
            & 3      & 0.0093 
& 0.0011
& 88.74 & 0.8433 & 0.9921 & 17.64 \\
            & 4      & 0.0219 
& 0.0014
& 93.51 & 0.6166 & 0.9946 & 61.30 \\
            & 5      & 0.0401 
& 0.0051
& 87.41 & 0.9029 & 0.9961 & 10.32 \\
            & 6      & 0.0411 
& 0.0049
& 87.99 & 0.7564 & 0.9951 & 31.55 \\
            & 7      & 0.1453 
& 0.0164
& 88.73 & 0.9245 & 0.9960 & 7.74 \\
            & 8      & 0.1492 
& 0.0153
& 89.72 & 0.8710 & 0.9964 & 14.40 \\
            & 9      & 0.1367 
& 0.0135
& 90.11 & 0.8883 & 0.9972 & 12.26 \\
            & 10     & 0.2005 
& 0.0153
& 92.37 & 0.8473 & 0.9976 & 17.74 \\
            & Reward & 0.0280 & 0.0024& 91.50 & 0.9671 & 0.9999 & 3.39 \\

Walker2d    & 0      & 0.0077 
& 0.0023
& 70.09 & 0.7090 & 0.9576 & 35.06 \\
            & 1      & 0.0525 
& 0.012
& 77.11 & 0.5628 & 0.9598 & 70.55 \\
            & 2      & 0.0451 
& 0.0109
& 75.78 & 0.5502 & 0.9674 & 75.84 \\
            & 3      & 0.0349 
& 0.0129
& 62.93 & 0.8363 & 0.9758 & 16.68 \\
            & 4      & 0.0475 
& 0.0222
& 53.16 & 0.8854 & 0.9833 & 11.06 \\
            & 5      & 0.0449 
& 0.0111
& 75.37 & 0.5612 & 0.9656 & 72.04 \\
            & 6      & 0.0342 
& 0.0118
& 65.51 & 0.8541 & 0.9791 & 14.63 \\
            & 7      & 0.0428 
& 0.0211
& 50.75 & 0.9016 & 0.9855 & 9.31 \\
            & 8      & 0.0881 
& 0.0265
& 69.99 & 0.8335 & 0.9806 & 17.65 \\
            & 9      & 0.1897 
& 0.0529
& 72.11 & 0.6446 & 0.9651 & 49.73 \\
            & 10     & 1.0088 
& 0.2438
& 75.83 & 0.7850 & 0.9815 & 25.02 \\
            & 11     & 1.0388 
& 0.3302
& 68.21 & 0.8802 & 0.9854 & 11.96 \\
            & 12     & 0.8285 
& 0.3426
& 58.65 & 0.9159 & 0.9837 & 7.40 \\
            & 13     & 1.4659 
& 0.882
& 39.83 & 0.8874 & 0.9539 & 7.50 \\
            & 14     & 0.9653 
& 0.3115
& 67.73 & 0.8907 & 0.9862 & 10.72 \\
            & 15     & 0.8146 
& 0.3546
& 56.47 & 0.9144 & 0.9847 & 7.68 \\
            & 16     & 1.3732 
& 0.7956
& 42.06 & 0.8975 & 0.9590 & 6.85 \\
            & Reward & 0.0296 & 0.0095& 67.80 & 0.9897 & 0.9992 & 0.95 \\
\bottomrule
\end{tabular}
}
\end{table*}

\clearpage
\subsection{Policy Performance}

\begin{figure}[!ht]
\centering
\begin{subfigure}{0.48\textwidth}
        \centering
        \includegraphics[width=\linewidth]
        {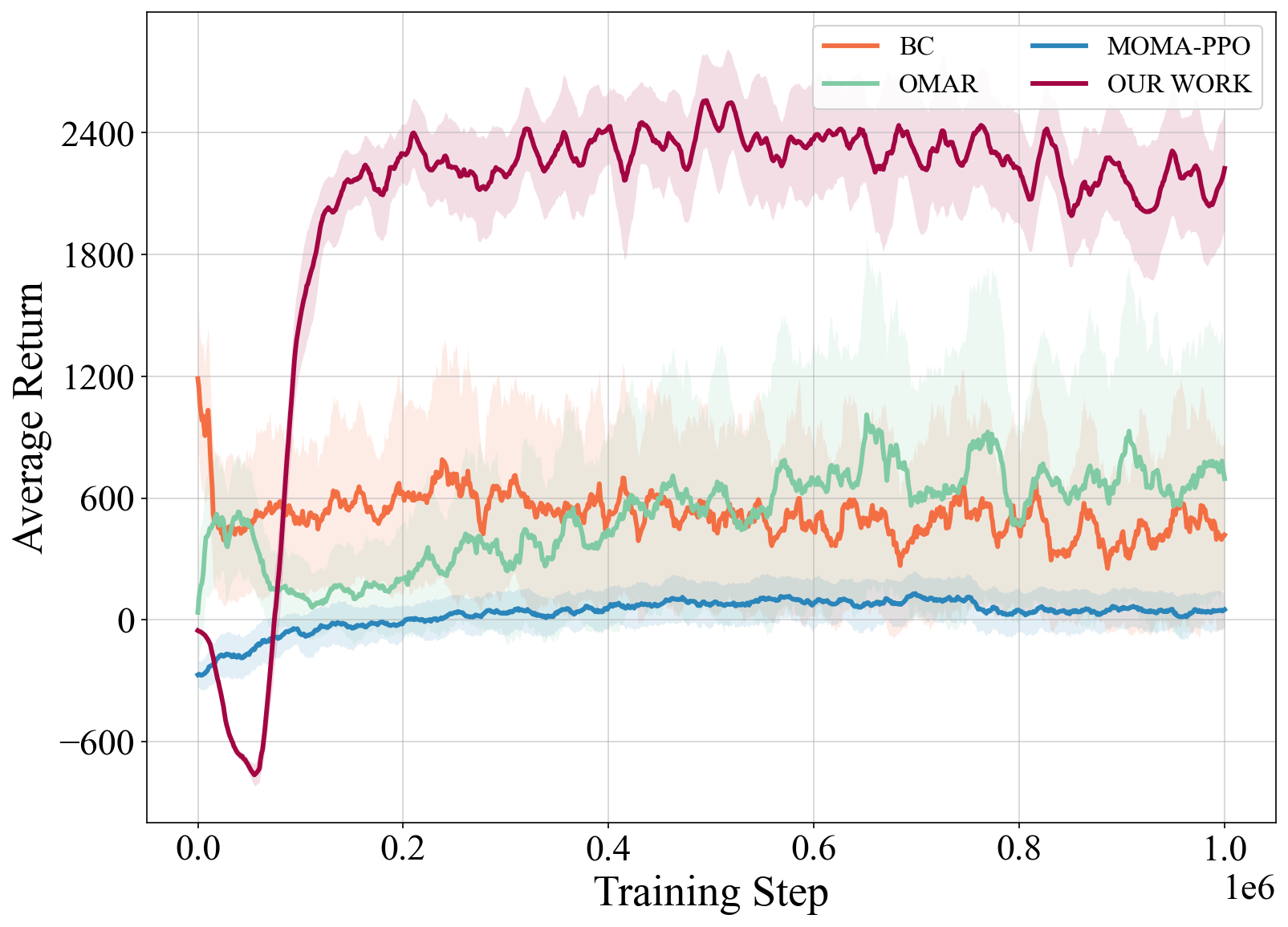}
        \caption{HalfCheetah Expert}
    \end{subfigure}
\hfill
\begin{subfigure}{0.48\textwidth}
        \centering
        \includegraphics[width=\linewidth]
        {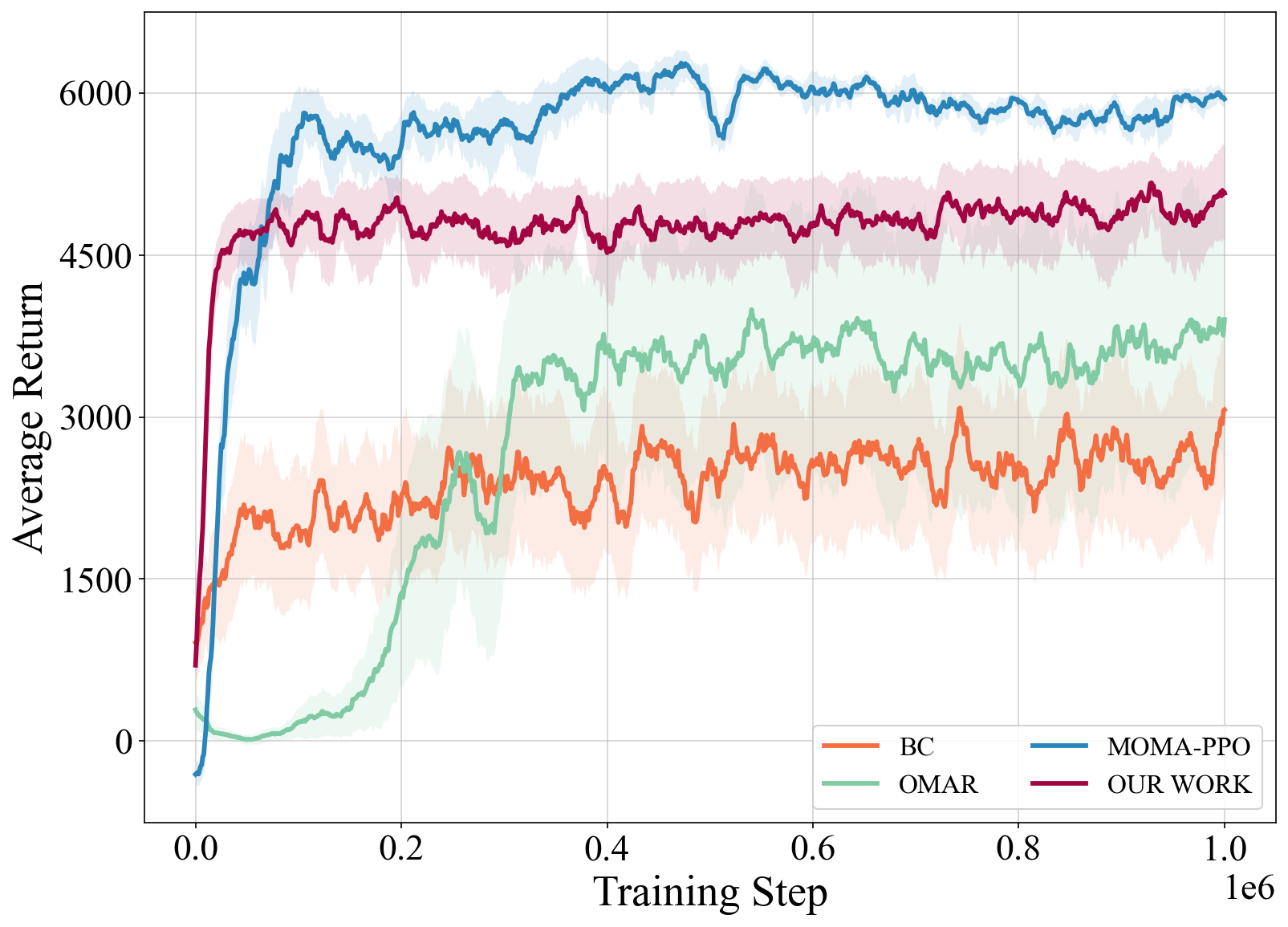}
        \caption{HalfCheetah Full}
    \end{subfigure}

\vspace{0.6em}

\begin{subfigure}{0.48\textwidth}
        \centering
        \includegraphics[width=\linewidth]
        {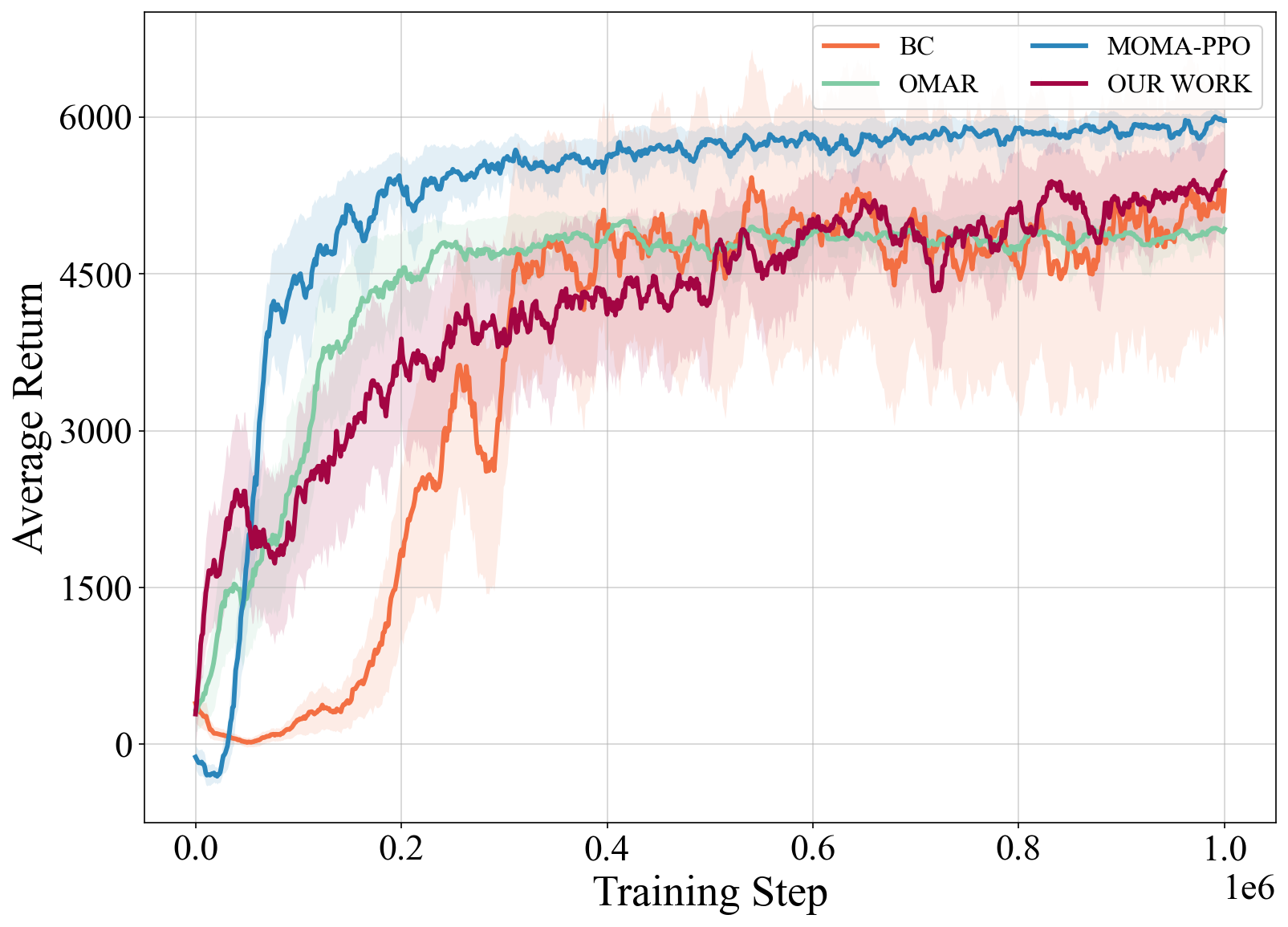}
        \caption{HalfCheetah Medium}
    \end{subfigure}
\hfill
\begin{subfigure}{0.48\textwidth}
        \centering
        \includegraphics[width=\linewidth]
        {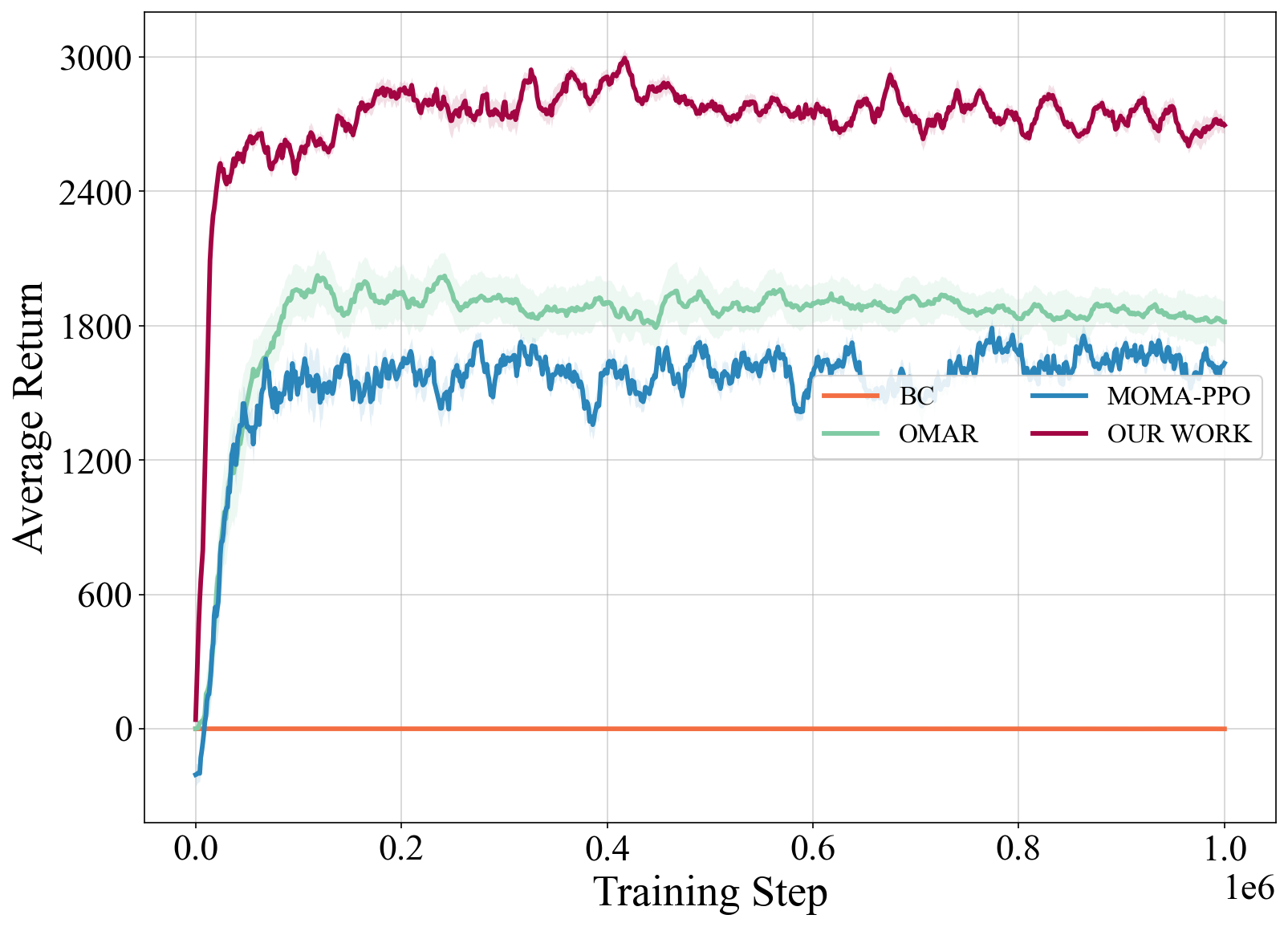}
        \caption{HalfCheetah Random}
    \end{subfigure}

\vspace{0.6em}

\begin{subfigure}{0.48\textwidth}
        \centering
        \includegraphics[width=\linewidth]
        {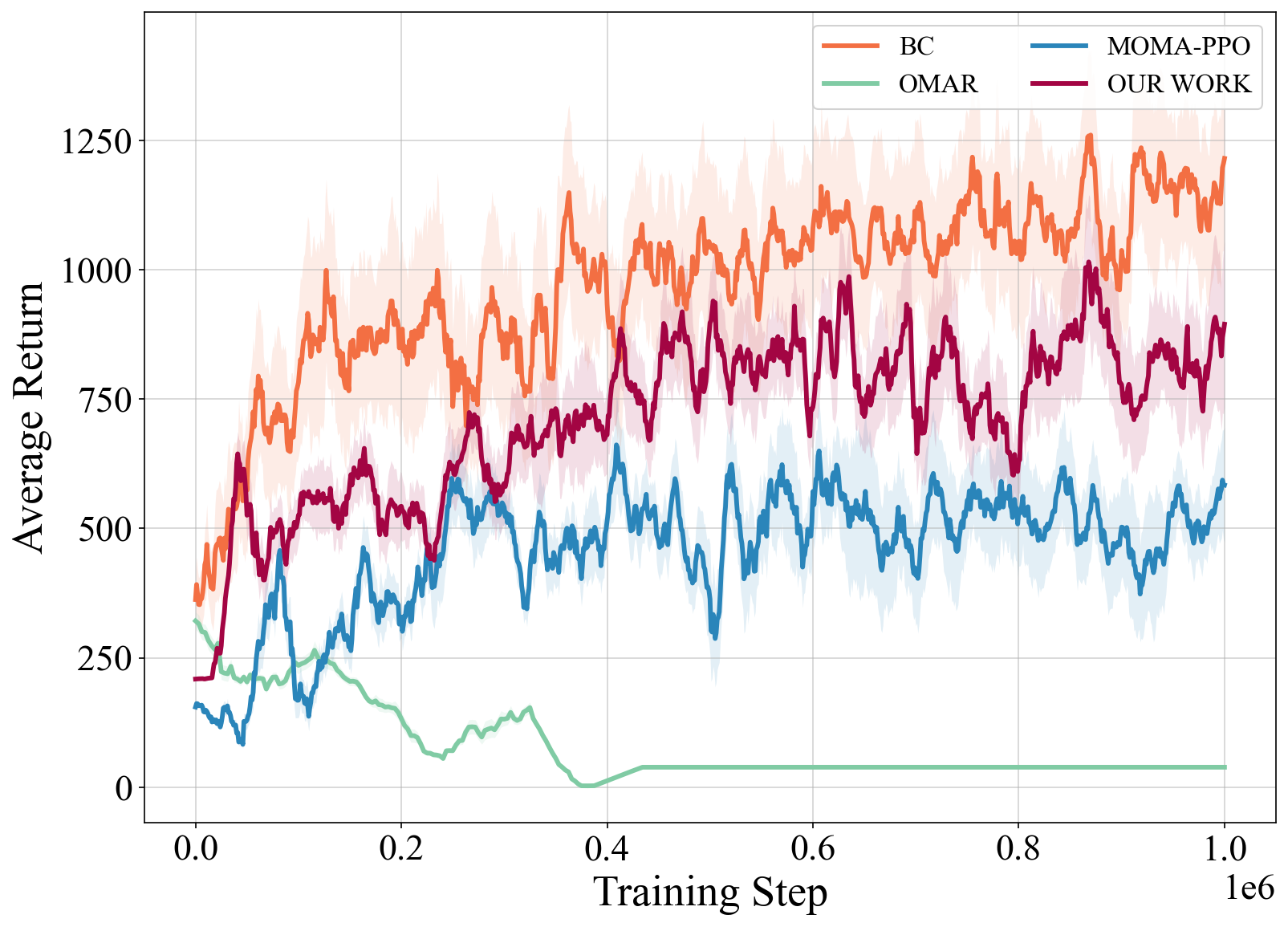}
        \caption{Hopper Expert}
    \end{subfigure}
\hfill
\begin{subfigure}{0.48\textwidth}
        \centering
        \includegraphics[width=\linewidth]
        {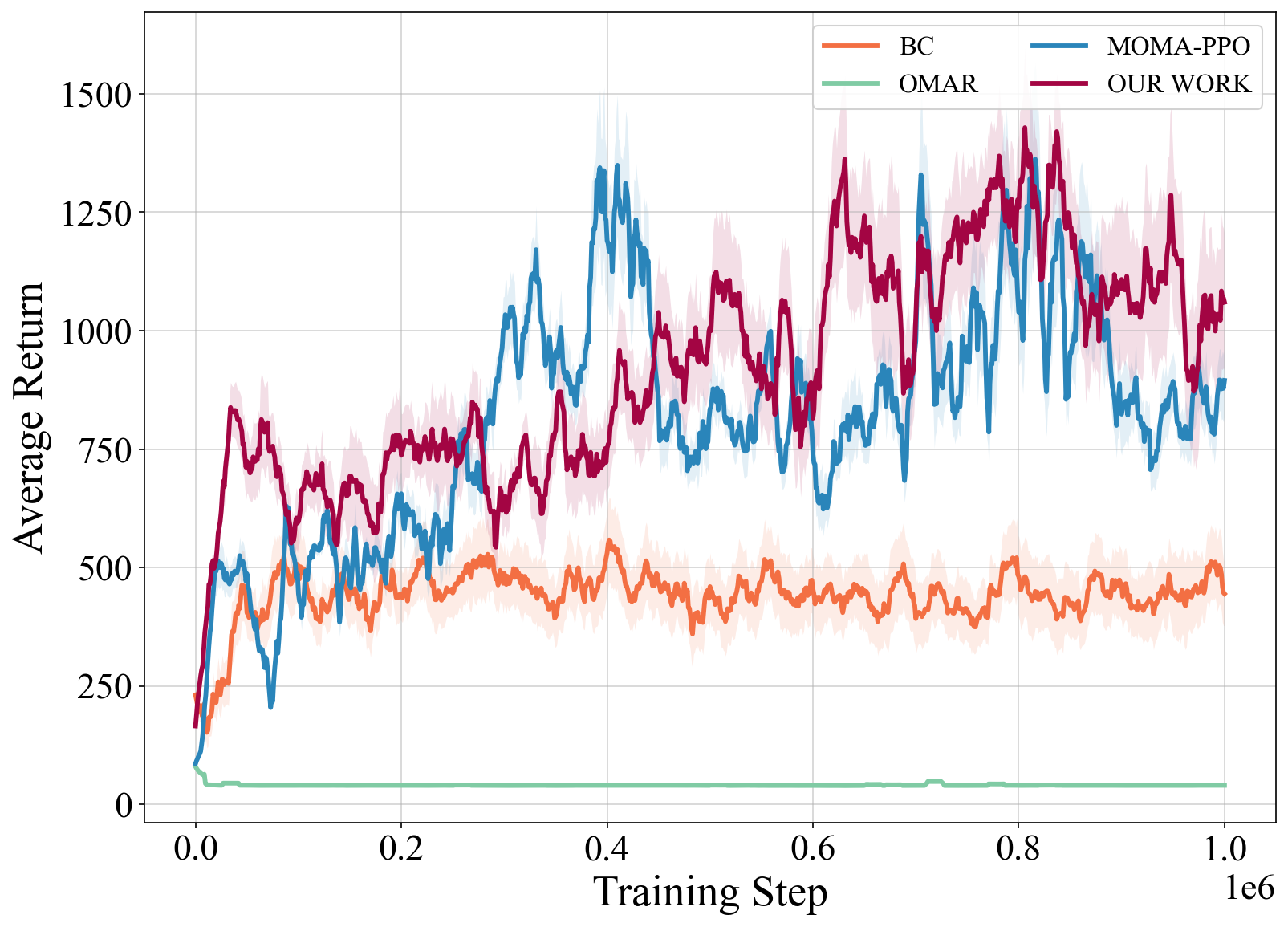}
        \caption{Hopper Full}
    \end{subfigure}
\caption{Performance comparison across different environments and dataset qualities (panels a--f).}
\label{fig:all_results}
\end{figure}

\clearpage
\begin{figure}[!t]
\ContinuedFloat
\centering
\begin{subfigure}{0.48\textwidth}
        \centering
        \includegraphics[width=\linewidth]
        {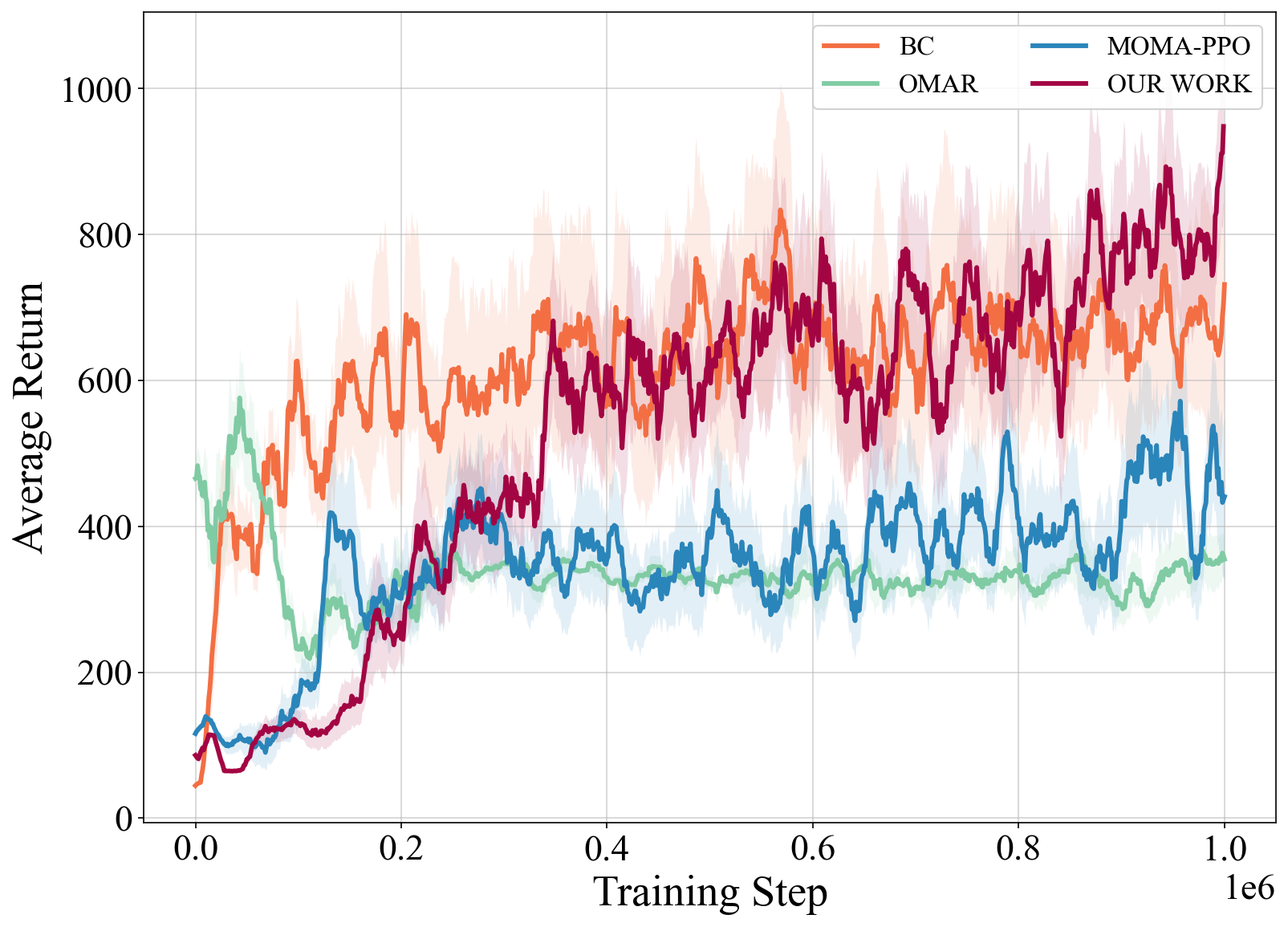}
        \caption{Hopper Medium}
    \end{subfigure}
\hfill
\begin{subfigure}{0.48\textwidth}
        \centering
        \includegraphics[width=\linewidth]
        {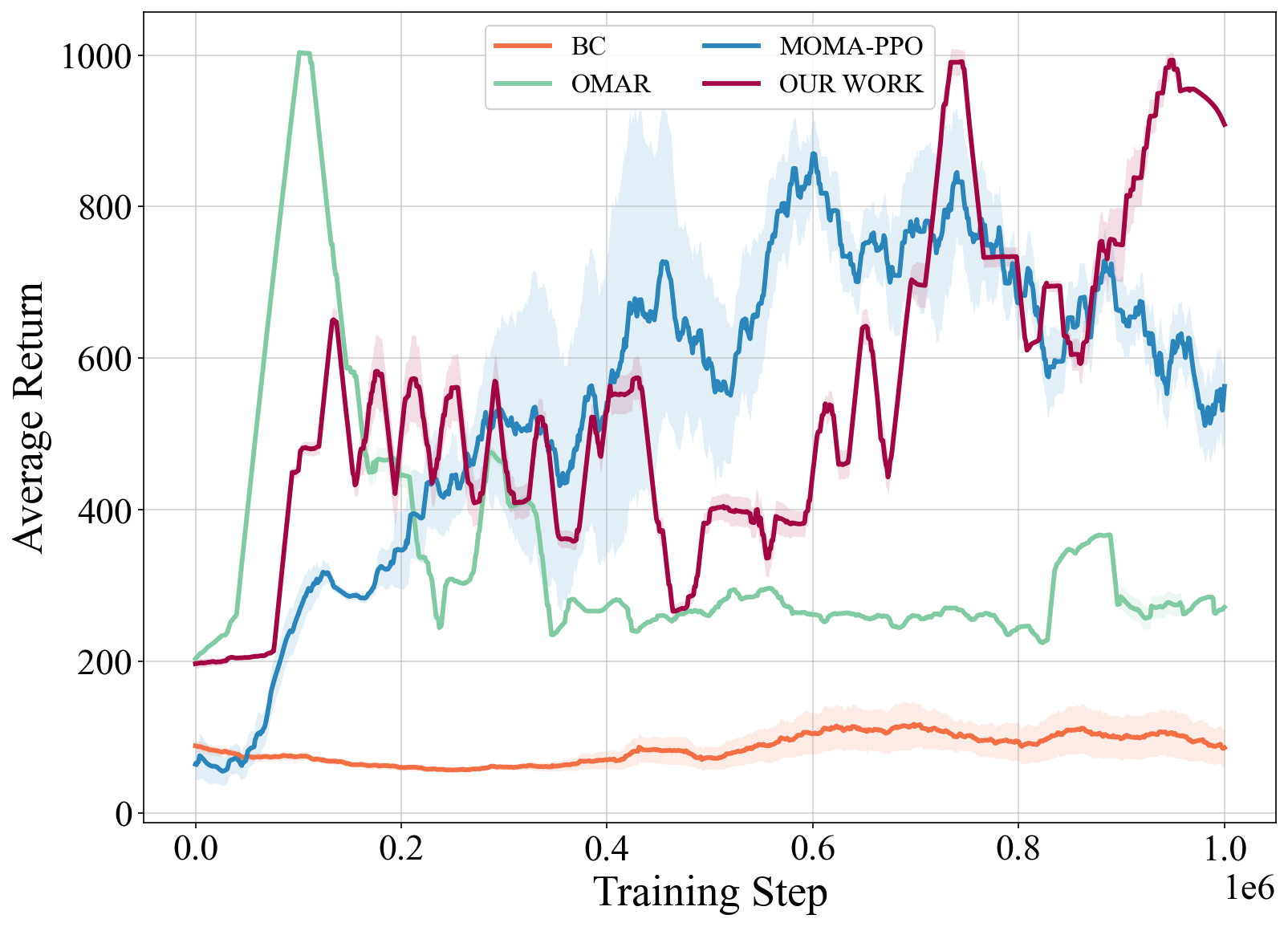}
        \caption{Hopper Random}
    \end{subfigure}

\vspace{0.6em}

\begin{subfigure}{0.48\textwidth}
        \centering
        \includegraphics[width=\linewidth]
        {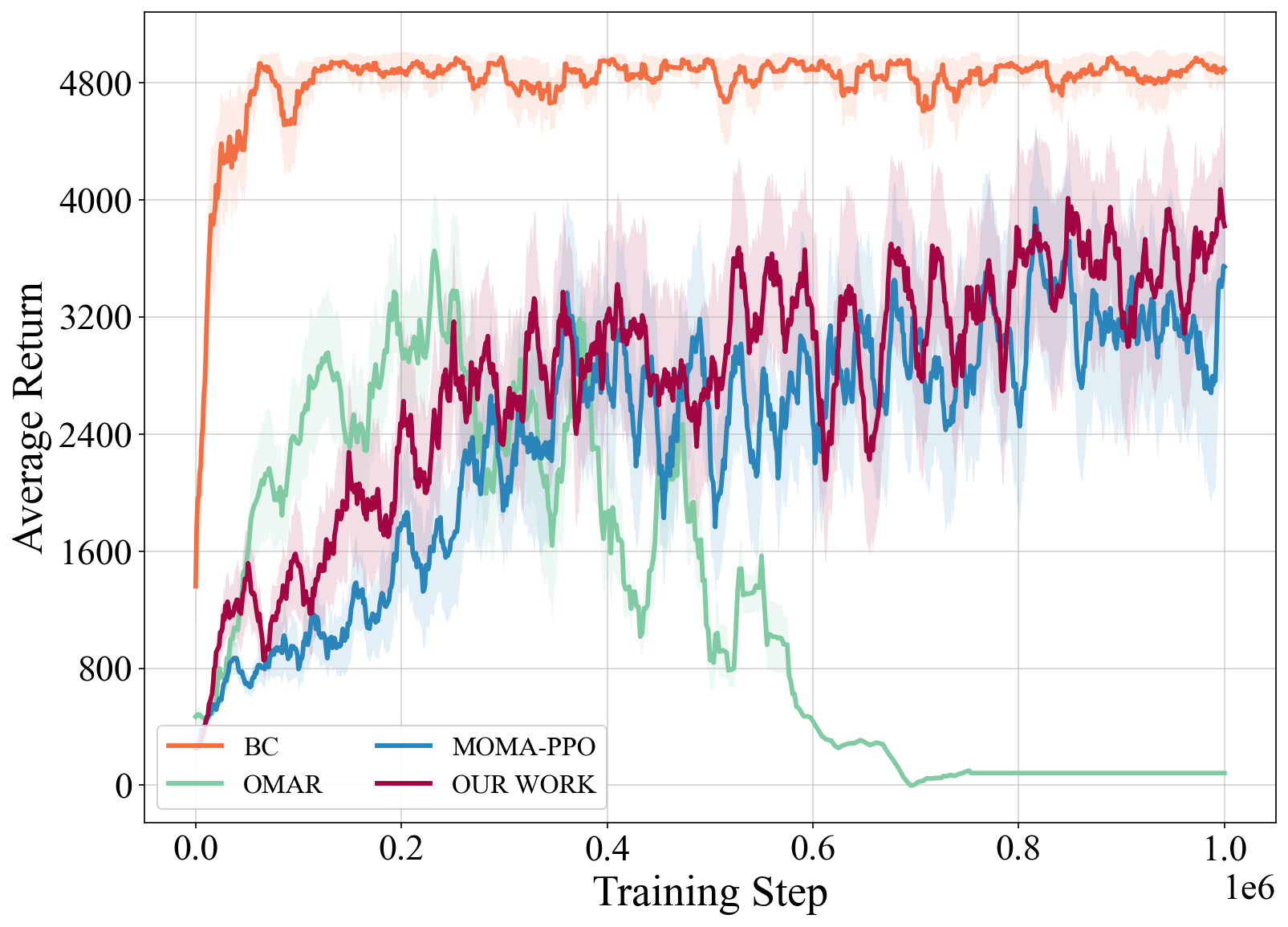}
        \caption{Walker2d Expert}
    \end{subfigure}
\hfill
\begin{subfigure}{0.48\textwidth}
        \centering
        \includegraphics[width=\linewidth]
        {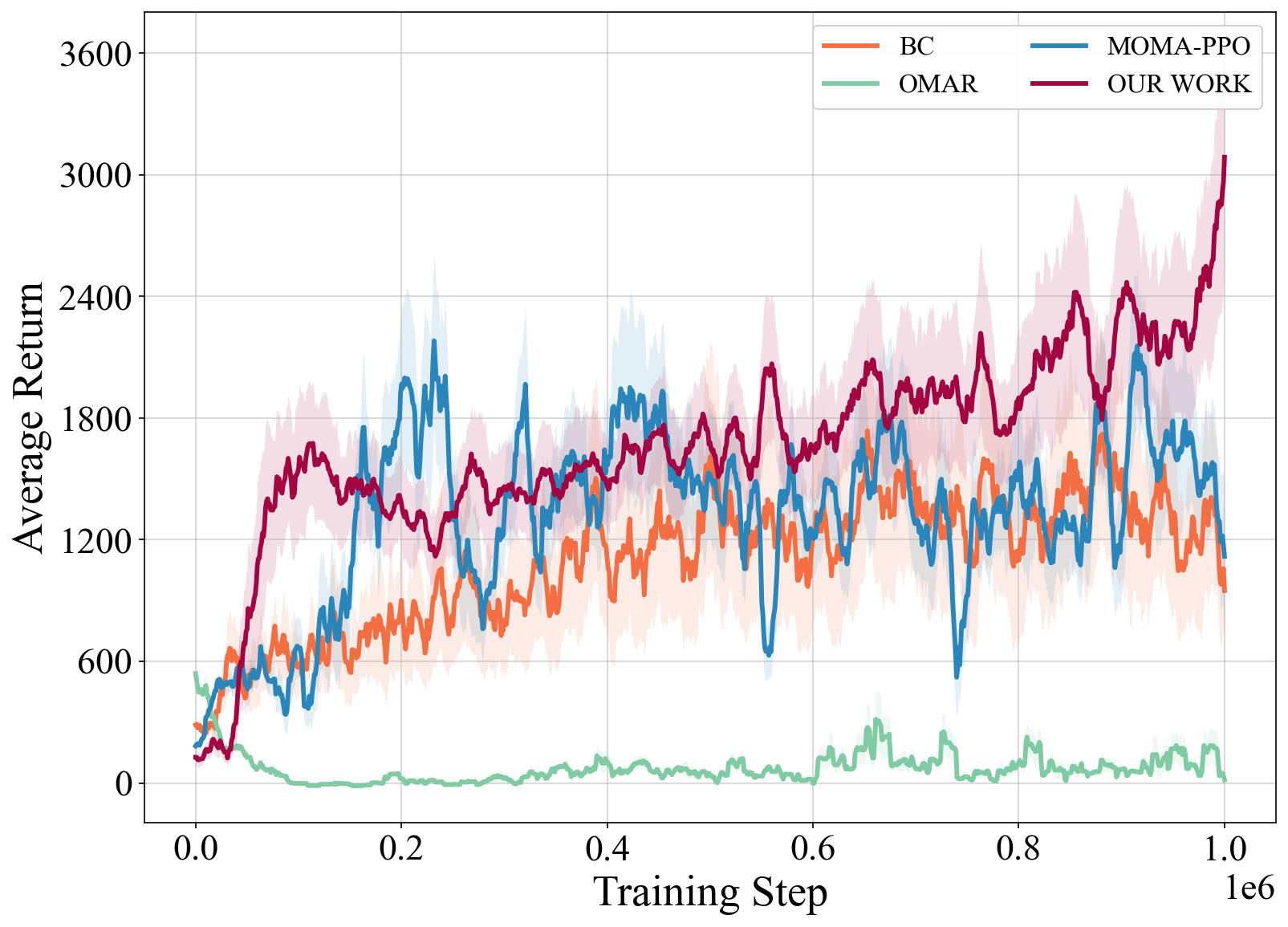}
        \caption{Walker2d Full}
    \end{subfigure}

\vspace{0.6em}

\begin{subfigure}{0.48\textwidth}
        \centering
        \includegraphics[width=\linewidth]
        {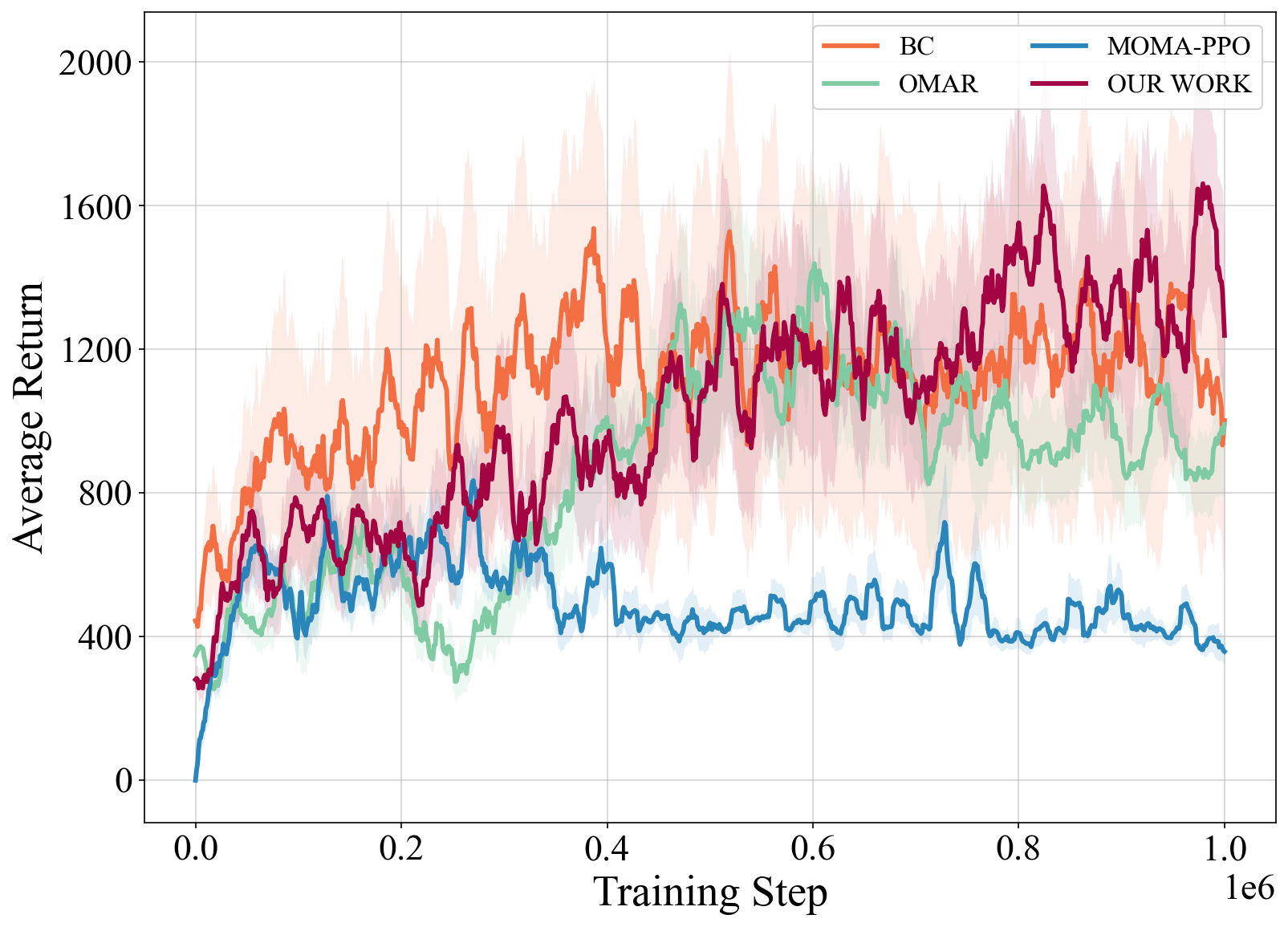}
        \caption{Walker2d Medium}
    \end{subfigure}
\hfill
\begin{subfigure}{0.48\textwidth}
        \centering
        \includegraphics[width=\linewidth]
        {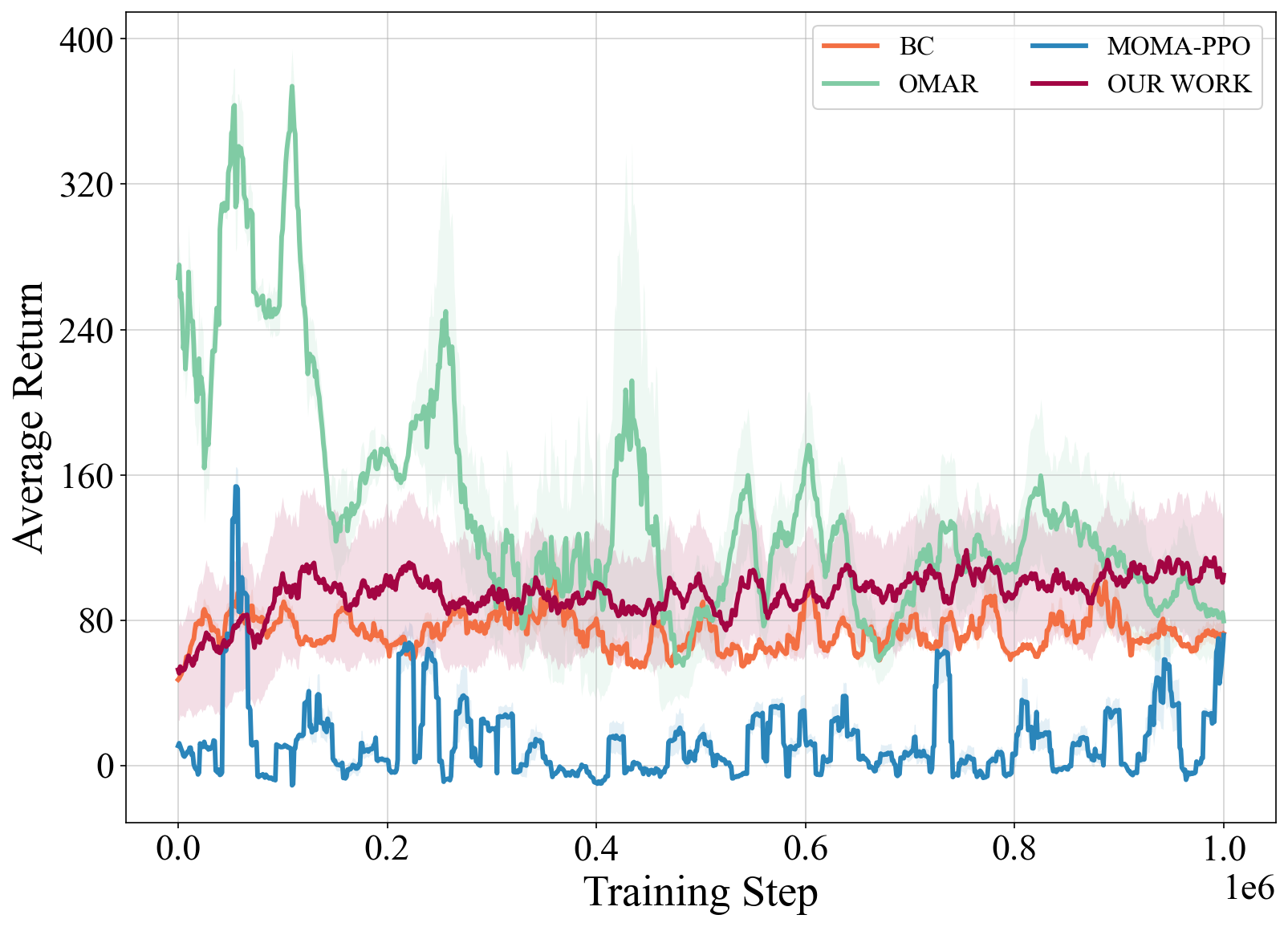}
        \caption{Walker2d Random}

    \end{subfigure}
\caption{Performance comparison across different environments and dataset qualities (continued; panels g--l).}
\end{figure}